\documentclass[]{interact}

\usepackage{epstopdf}% To incorporate .eps illustrations using PDFLaTeX, etc.
\usepackage{subfigure}% Support for small, `sub' figures and tables
\usepackage[british]{babel}
\usepackage{natbib}% Citation support using natbib.sty
\bibpunct[, ]{(}{)}{;}{a}{}{,}% Citation support using natbib.sty
\theoremstyle{plain}% Theorem-like structures provided by amsthm.sty
\newtheorem{theorem}{Theorem}[section]
\newtheorem{lemma}[theorem]{Lemma}

\theoremstyle{definition}
\newtheorem{definition}[theorem]{Definition}

\theoremstyle{remark}
\newtheorem{remark}{Remark}

\usepackage{verbatim}
\usepackage{amsmath,amssymb,calc}
\usepackage{graphicx,array}
\usepackage{enumerate}
\usepackage{algorithm}
\usepackage{color}
\usepackage{algpseudocode}
\usepackage{tikz}
\usetikzlibrary{decorations.pathreplacing}
\newcommand{\beq}{\begin{equation}}
\newcommand{\eq}{\end{equation}}
\newcommand{\E}{\mathbb{E}}

\renewcommand{\P}{\mathbb{P}}
\newtheorem{proofpart}{Segment}

\makeatletter
\@addtoreset{proofpart}{theorem}
\makeatother

\DeclareMathOperator*{\argmin}{arg\,min}
\title{Platoon Forming Algorithms for Intelligent Street Intersections}

\author{
	\name{R.W. Timmerman\textsuperscript{a}\thanks{Contact R.W. Timmerman, email: r.w.timmerman@tue.nl} and M.A.A. Boon\textsuperscript{a}}
	\affil{\textsuperscript{a}Eindhoven University of Technology, Eindhoven, The Netherlands}
}

\begin{document}

	\maketitle
	
	\begin{abstract}
We study intersection access control for autonomous vehicles. Platoon forming algorithms, which aim to organize individual vehicles in platoons, are very promising. To create those platoons, we slow down vehicles before the actual arrival at the intersection in such a way that each vehicle can traverse the intersection at high speed. This increases the capacity of the intersection significantly, offering huge potential savings with respect to travel time compared to nowadays traffic.

We propose several new platoon forming algorithms and provide an approximate mean delay analysis for our algorithms. A comparison between the current day practice at intersections (through a case study in SUMO) and our proposed algorithms is provided. Simulation results for fairness are obtained as well, showing that platoon forming algorithms with a low mean delay sometimes are relatively unfair, indicating a potential need for balancing mean delay and fairness.
	\end{abstract}

	
	\section{Introduction}\label{h:intro}

Congestion is commonplace at intersections in urban traffic, but intersections are inevitable to divide capacity among vehicles from conflicting flows. To do so in a fair and efficient manner, intersections are typically managed by some kind of switching process that alternatingly gives access to batches of vehicles, imposing a constraint on the maximal batch size that can pass the intersection.

The traditional way of regulating the switching process is by installing traffic lights with static signalling, using timers, see e.g.~\cite{darroch} and~\cite{fctlsolo}, or dynamic signalling with sensor data of currently existing traffic flows, see e.g.~\cite{papageorgiou2003review}. Anticipating the emergence of self-driving vehicles, efficient and fair algorithms for intersection access should be designed. Platoon Forming Algorithms (PFAs) provide such alternatives for self-driving vehicles, no longer letting the traffic lights dictate the switching process and hence batch forming, but letting the vehicles organize themselves in batches, well in advance of arriving at the intersection as in~\cite{miculescu2014polling,miculescu2016polling} and~\cite{tachet2016revisiting}. In this way platoons of vehicles are formed that can pass the intersection collectively.

There is a natural tension between capacity and fairness. One of the fairest switching rules is to let vehicles pass the intersection in order of arrival (on an intersection wide basis). This rapidly becomes unsustainable, because each switch requires an additional clearance time, which decreases the capacity of the intersection. In near-saturation conditions, when the flows together impose a high volume-to-capacity ratio, the loss of capacity due to switching will have a dramatic effect on delays. Our PFAs aim to balance capacity and fairness.

In PFAs, vehicles arriving at the intersection arrange themselves in platoons, not adapting their relative position to other vehicles on the same lane but adapting their speed. The key feature is that cars, while approaching the intersection, adjust their speeds and upon arrival at the intersection are at high speed, accessing the conflict area of the intersection for a minimum period of time. In this way, time bans to give way to other traffic flows still exist, but the platoons are processed in the quickest possible way, because the size and speed of the platoons, of all directions, are organized by the PFA.

PFAs are one particular example of the `slower is faster' effect, which is also observed in e.g.~\cite{helbing2000simulating} and~\cite{helbing2009operation}, where, perhaps counter-intuitively, slowing down early results in less delay on average in the future. Moreover, this phenomenon results in environmental advantages as less braking-and-pulling-up-again is needed and cars reach their destination more quickly.

The importance of intersection access algorithms has been recognized for several years. Examples of PFAs can be found in~\cite{tachet2016revisiting}, which introduces a batch formation algorithm based on arrival times of vehicles and a maximum batch size, and in~\cite{miculescu2014polling,miculescu2016polling}, which use an approach based on queueing models, and to be more specific \emph{polling models}. Queueing theory has played an important role in the modeling and performance analysis of signalized and unsignalized intersections since the early sixties, see for example the seminal papers by \citet{newell65} and \cite{tanner62}. Polling models are queueing systems where one server visits multiple queues, generally in a cyclic order. The multidimensional analysis of polling models is able to capture the intrinsic interactions between the processes taking place at the different queues. Polling models have a long tradition in communication networks, but \cite{miculescu2014polling} have shown that their applicability can be extended to intersections of autonomous vehicles. One of the key questions in polling models is how to decide which queue should be served (and how many customers before advancing to the next queue). This is exactly one of the main topics of this paper, where we develop algorithms that determine how to construct platoons of autonomous vehicles and when to give each platoon access to the intersection.
A speed profile algorithm provides the key link between the PFAs and polling models, which we will show in more detail later.
In the existing literature, many more models and control techniques are investigated like reservation based control algorithms \citep{dresner2008multiagent} and controls based on fuzzy logic \citep{milanes2010controller}. In a recent paper, \cite{liuwanghoogendoorn2019} present a method for fixed-cycle plans, where the PFA and the optimization of the trajectories are integrated. For an overview we refer to~\cite{rios2017survey}.

The area of application of PFAs is not restricted to intersections. There are numerous cases where PFAs could be used to achieve a good performance. An example in traffic would be the merging of different streams of vehicles (discussed in e.g.~\citealt{rios2017survey}). Another possible application can be found in automated guided vehicles (AGVs) systems, where AGVs cross each other or have to merge, see e.g.~\cite{kockelkoren2018}. In~\cite{kockelkoren2018}, ideas stemming from speed profile algorithms are used and so PFAs can be used in similar types of AGV systems.
%! add something about aviation/planes?

\subsection*{Main Contributions}

Our first contribution is the introduction of several new PFAs, based on enhanced polling policies, that perform well regarding mean delay, unifying and extending ideas from~\cite{miculescu2016polling} and~\cite{tachet2016revisiting}.

Our second contribution is the introduction of a new class of closed-form speed profile algorithms (SPAs). SPAs ensure an efficient use of the intersection, by optimizing the trajectory of (platoons of) vehicles driving towards the intersection, ensuring the arrival at their designated times.
\cite{miculescu2014polling} introduced the MotionSynthesize procedure, a linear optimization program to find these trajectories. The MotionSynthesize algorithm computes the optimal trajectory for an autonomous vehicle, given the trajectory of its predecessor and the crossing time computed by the PFA. We have developed an alternative to this procedure exhibiting desirable properties and we have found closed-form solutions for the MotionSynthesize procedure and for this alternative, alleviating the need for linear optimization solvers.

Using such speed profile algorithms, a link between polling models and PFAs is established, making it possible to conduct a performance analysis on e.g. mean delay, which is the main performance characteristic considered in the literature for algorithms like PFAs. Using interpolation techniques from \cite{boon2011closed} we develop accurate approximations for the mean delays.

Another contribution is the introduction of the notion of \emph{fairness} of a PFA. Fairness in queueing models (and therefore PFAs) is important in the perception of customers, see e.g.~\cite{rafaeli2002effects}. We use the quantification of fairness as defined in~\cite{shapira2015fairness} for polling models, to assess the fairness of the various PFAs.

%As there exists a trade-off between fairness and mean delay, we optimize a weighted sum of the two quantities to account for this trade-off. Approximations for the mean delay and fairness are required to do so and we derive those approximations for our PFAs.

Furthermore, we provide a comparison between the performance of traditional traffic technologies and PFAs through simulations in SUMO and show that intersections in the future can be used much more efficiently, reducing congestion.

\subsection*{Organization of the paper}

This paper is organized in the following way. We start with a description of the various ingredients of the model and provide an extensive description of the new PFAs that we introduce in Section~\ref{h:model}. Section~\ref{h:arrivalAtI} is devoted to SPAs. Afterwards, we introduce polling models and give the analytical results that we need for the analysis of mean delay and fairness of PFAs in Section~\ref{h:analysis}. Subsequently, Section~\ref{h:comparison} provides a comparison between the traditional traffic light (represented by simulations in SUMO) and our PFAs, focusing on mean delay, and we wrap up with conclusions and discussions in Section~\ref{h:conclusion}.

\section{Model Formulation}\label{h:model}

We will consider models in which autonomous vehicles are crossing an intersection. We assume the existence of a control region around the intersection with at the center a centralized controller communicating with all vehicles within the control region. In fact, this control region can be divided into two sub-regions: the inner part is called the ``SPA control region''. As soon as a vehicle enters this part of the control region, its trajectory is determined by the speed profiling algorithm. In the outer part, which we call the ``PFA control region'', the access times of each of the arriving vehicles to the intersection are determined. In this PFA control region, the central controller creates platoons of vehicles by scheduling the crossing times of the vehicles according to some policy (the PFA) in such a way that every vehicle is able to cross the intersection at its designated time. We assume that we can control the speed of a vehicle and do so in such a way that the intersection is used efficiently. We make sure that vehicles drive at maximum speed at the moment they are starting to cross the intersection, using ideas introduced in~\cite{miculescu2014polling}. Instead of stopping at the stop line and still having to accelerate when crossing the intersection, a vehicle is already slowed down before it reaches the intersection and starts accelerating again, such that it is driving at full speed when reaching the conflict area of the intersection. This amongst others implies that the time to cross the intersection is the same for each vehicle. The last assumption discussed here, is that we assume that the central controller can look `ahead' for the same amount of time for each of the lanes, to ease the notation and algorithms. The reason why we need separate control regions for the PFA and the SPA is that we need the trajectory to be fixed once a vehicle enters the SPA control region. Inside the PFA control region, vehicle access times may be adjusted due to the arrival of other vehicles. %! Name V2I communication?

We clarify how this works in a simple example, depicted in Figure \ref{fig:intersection}. For simplicity, we show vehicles arriving from only two different approaches (marked red and blue). The central controller uses a PFA to compute the access times to (the conflict area of) the intersection for each vehicle entering the control region. The intersection drawn in Figures \ref{fig:intersection}(a) and (b) only depicts the inner (SPA) part of the control region. Figure \ref{fig:intersection}(c) shows the corresponding trajectories. Note that all vehicles drive at full speed in the PFA control area (from 75 -- 50 metres distance) and start their trajectories controlled by the SPA at 50 metres distance. %The two parts of the control region are separated by a grey line.
The blue vehicle entering the SPA control region at time $t=0$ encounters no hinder from other vehicles and proceeds at full speed, without delay. The first red vehicle was originally scheduled to arrive at the intersection directly after the first blue vehicle. When, however, the second blue vehicle entered the PFA control region at $t=1$ (probably arriving in a platoon from an upstream intersection), this blue vehicle is allowed to join the platoon started by the previous blue vehicle. This means that the first red vehicle is rescheduled, being delayed for two seconds. This means that it gets access to the intersection \emph{after} the second blue vehicle, at a safe distance. Due to this two second delay, the next two red vehicles are able (and allowed) to join the red platoon. The actual trajectories towards the intersection are determined by an SPA, which ensures an efficient usage of the intersection. Note that all vehicles cross the intersection at full speed. %This process keeps on repeating itself.

%\definecolor{lgray}{rgb}{0.9,0.9,0.9}

\begin{figure}[!ht]
\begin{tabular}{cc}
\includegraphics[width=0.4\linewidth]{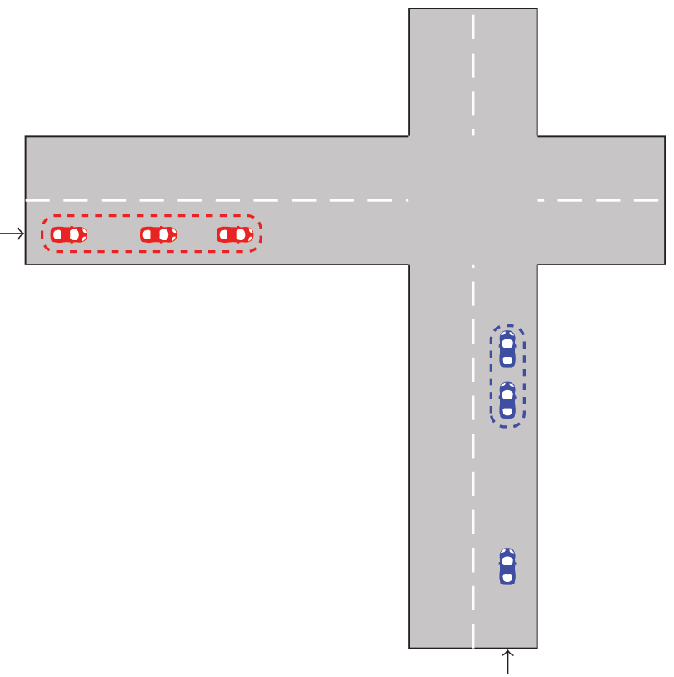}
&
\includegraphics[width=0.4\linewidth]{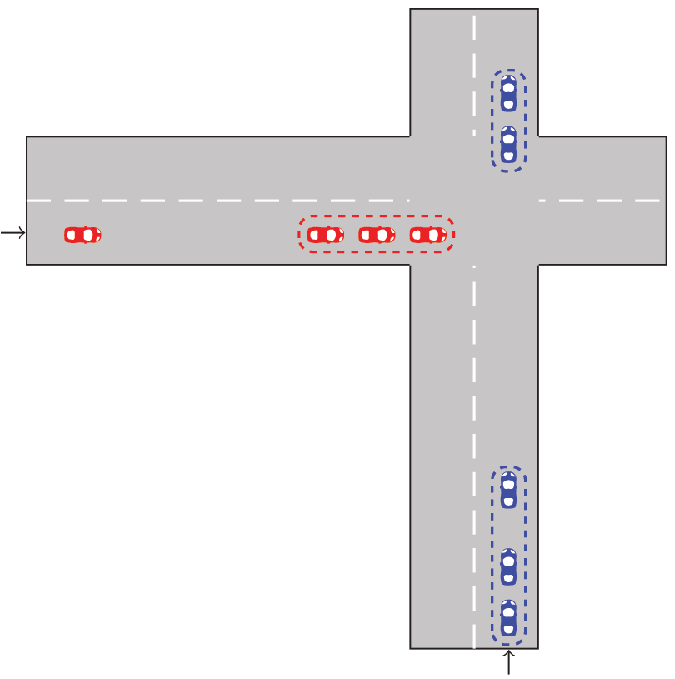}
\\
\scriptsize (a) %Platoon forming approaching the intersection
Intersection at time $t=4$.
&
\scriptsize (b) %Crossing the intersection in platoons
Intersection at time $t=8$.
\end{tabular}
\centering
\includegraphics[width=0.8\textwidth]{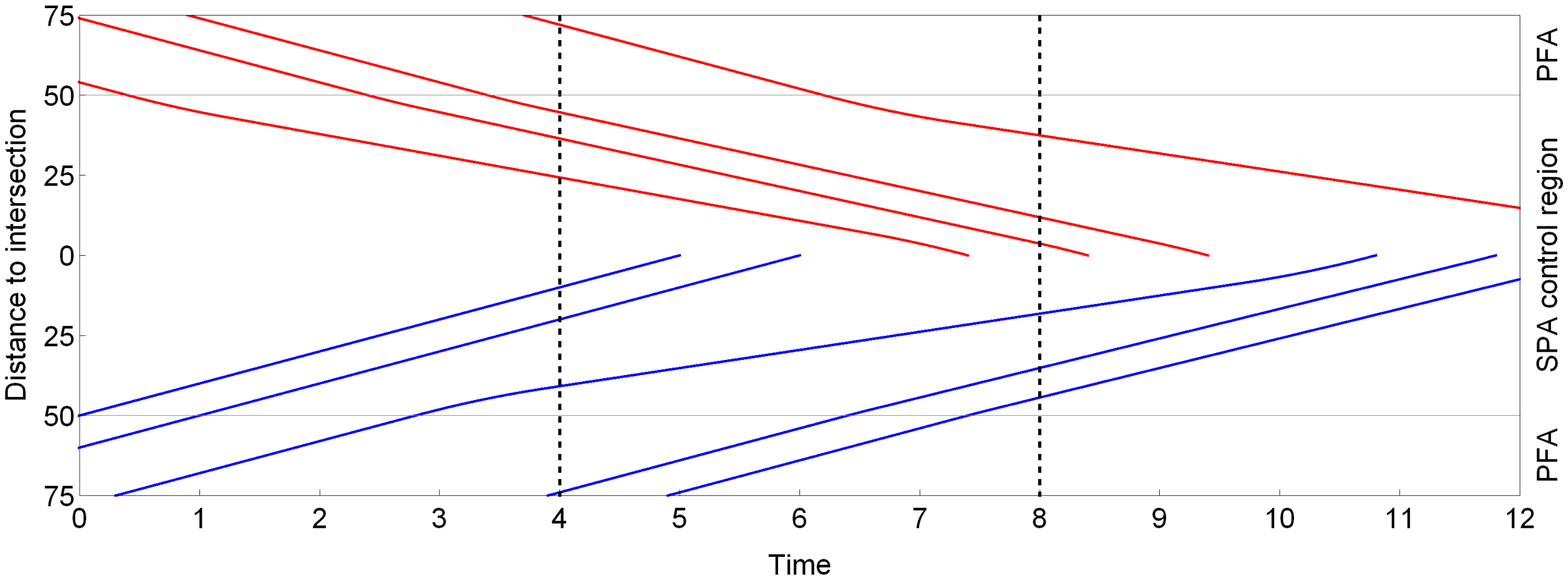}\\
\scriptsize (c) Trajectories for the red and blue traffic flows.
\caption{A schematic representation of the model discussed in this paper. The platoon forming algorithms in this paper determine how the platoons are constructed. In the next step, a speed profiling algorithm determines how each individual vehicle approaches the intersection. Fig (a) and (b) correspond, respectively, to the situation in (c) at times $t=4$ and $t=8$ seconds.}
\label{fig:intersection}
\end{figure}

An advantage of the control region, besides the ability to control the speed of arriving vehicles, is that we can adjust the scheduling of the vehicles based on the arriving vehicles that are not yet at the intersection. This specific anticipation is key to the forming of platoons and is up to the central controller at the intersection and results in a specific PFA. There are many PFAs, yet we will specifically focus on PFAs that find their origin in polling models, because they are efficient, well understood and have proven their value in other application areas, such as communication systems and production lines.

\subsection{Platoon Forming Algorithms}\label{h:pfas}

We present our new PFAs as standalone algorithms, based on service disciplines for polling models, which are described in a way fit for PFAs. We also briefly discuss the Batch Algorithm, originating from~\cite{tachet2016revisiting}, which serves as a benchmark for our PFAs.
The PFAs we discuss, are all derived from so-called branching-type disciplines, which find their origin in the polling literature, see e.g.~\cite{resing1993polling}. Branching-type service disciplines include the exhaustive and the gated discipline, which all allow for many analytical results.

Before we start with the descriptions of the PFAs we introduce some concepts and notation. The PFA determines the crossing time of each of the vehicles in the control region that have not yet crossed the intersection. We represent this schedule by an entity we call `vehicles'. A vehicle $V$ has three properties: a lane $d_{V}$, an earliest crossing time $a_{V}$ and the currently scheduled crossing time $c_{V}$. We assume that at every point in time we have such a list of vehicles, ordered on basis of the $c_{V}$'s. The PFA updates (part of) the crossing time of the vehicles upon arrival and departure epochs of vehicles in the control region. The latter is dealt with in an easy way: if the current time is $c_V+B$, where $B$ denotes the difference in crossing times between two vehicles on the same lane, then vehicle $V$ is removed from the ordering. Turning towards arrivals of vehicles within the control region, we need to consider the crossing times of all vehicles already scheduled in order to schedule $V$. There are several ways to schedule those vehicles and the first we discuss is the \emph{exhaustive discipline}, as described in Algorithm~\ref{alg:exhaustive}.
An intuitive explanation of the exhaustive discipline is the following: if a vehicle that arrives in the control region is able to get within $B$ seconds of the vehicle in front of it on the same lane (which might occur if the vehicle is delayed by its predecessor), it is allowed to join the same platoon as its predecessor. This would imply that all vehicles on different lanes have to wait an additional $B$ seconds. If a vehicle cannot join the platoon in front of it, it will form a new platoon. If no vehicle (on the current lane) is able to join the platoon currently crossing the intersection, the next platoon of vehicles at the next lane may cross the intersection. As a result we have a cyclic structure of departures of platoons.

This discipline is known for its low mean delay, which is the main reason to consider this discipline. We also introduce one more constant, $S$, that represents the time between the start of crossing of two vehicles on different lanes (similar to clearance times at intersections nowadays).

\begin{algorithm}
	\caption{exhaustive algorithm. }\label{alg:exhaustive}
	\begin{algorithmic}[1]
		\State Input: current ordering of vehicles, denoted $(V_{1},V_{2},...,V_{K})$, ordered on basis of $c_{V}$; $V_{last}$, defined as $V_K$ or the last vehicle that crossed the intersection if the ordering is empty; and a to be scheduled vehicle $V_0$ with earliest arrival time at the intersection $a_{V_0}$ in lane $d_{V_0}$.
	\If {$c_{V_{last}}+B<a_{V_0}$} \Comment{\small $V_0$ is scheduled last}
	\If {$d_{V_0}=d_{V_{last}}$}
	\State Put $c_{V_0}\gets a_{V_0}$. \Comment{\small $V_0$ proceeds without delay}
	\Else
	\State Put $c_{V_0}\gets \max\{a_{V_0},c_{V_{last}}+S\}$.
\Comment{\small Check if additional clearance time is needed}
	\EndIf
	\Else
%	\State Find for all lanes the last vehicle, $L_i$, $i=1,...,n$, currently ordered. Put $t_i\gets c_{L_i}$ if there is such a vehicle on lane $i$, otherwise put $t_i\gets -\infty$.
	\State Put $t_i\gets\begin{cases}
c_{L_i} & \text{ where $L_i$ is last scheduled vehicle in lane $i$},\\
-\infty & \text{ if lane $i$ is empty and no such vehicle exists.}
\end{cases}$
	\If {$t_{d_{V_0}}+B>a_{V_0}$} \Comment{\small $V_0$ is able to join a platoon}
	\State Put $c_{V_0}\gets t_{d_{V_0}}+B$.
	\For {each vehicle $V$ in the ordering with $c_V>t_{d_{V_0}}$ }
        \State Put $c_V\gets c_V+B$. \Comment{\small Delay other vehicles}
        \EndFor
	\Else
	\For { $l$ in $(d_{V_0}-1,d_{V_0}-2,...,1,n,n-1,...,d_{V_0}+1)$}
	\If{$t_{l}+S>a_{V_0}$}\Comment{\small $V_0$ starts new platoon after last platoon in lane $l$}
	\State Put $c_{V_0}\gets t_{l}+S$.
	\For {each vehicle $V$ in the ordering with $c_V>t_{l}$}
        \State Put $c_V\gets c_V+S$.  \Comment{\small Delay other vehicles}
    \EndFor
	\State \textbf{break}
	\EndIf
	\EndFor
	\EndIf
	\EndIf
	\State Add vehicle $V_{0}$ to the ordering.
	\State Output: the new ordering $(V_1,V_2,...,V_0,...,V_K)$%, where the vehicles $V$ are ordered on basis of $c_V$.
	\end{algorithmic}
\end{algorithm}

Although the exhaustive PFA will have very good delay characteristics, we will consider the \emph{gated PFA} (discussed below) as well.
The intuitive explanation of the gated algorithm is quite close to that of the exhaustive discipline, with one exception. It is not always allowed to join a platoon, even if a vehicle is able to get within $B$ seconds from its predecessor on the same lane. As described in more detail below, platoons are finalized at an earlier moment than with exhaustive service. This moment of finalizing a platoon is, in the polling literature, compared to putting a gate behind the last customer (corresponding to the last vehicle in the platoon). Newly arriving customers will have to wait (behind the virtual gate) for the next server visit, which corresponds to the formation of a new platoon in our setting.

An advantage of the gated discipline is that there is less variation in the size of platoons and, hence, cycle lengths are less variable as well. It may result in longer delays though, as we will see in the numerical examples in Section~\ref{h:analysis}.

For the implementation of the gated PFA, we need to keep track of a couple of additional variables for each lane. In this gated discipline we are namely `putting gates' which can be seen as `fixing the vehicles of a platoon', meaning that future arrivals in the same lane cannot join the currently formed platoon (i.e. they are `behind the gate').  We define two additional, ordered sets for each lane $f_i$ and $t_i$ representing the set of start times of platoons on lane $i$, respectively the end of platoons at lane $i$ (so the start of service of the last vehicle).
Joining a platoon is only allowed if the lane is \emph{not} the lane from which vehicles are currently departing (the platoon is not yet fixed). If a car in lane $i$ is able to reach the intersection (without any other interfering traffic) before one of times in $f_i$, then that car is allowed to join that platoon (so the platoon is enlarged). If such a car is not able to reach the intersection before one of the times in $f_i$, then it creates a new platoon. In general, departures of vehicles are dealt with in the same way as in the exhaustive discipline. We again have the cyclic structure as in the exhaustive discipline. The gated algorithm can then be described as in Algorithm~\ref{alg:gated}.

\begin{algorithm}
	\caption{gated algorithm.}\label{alg:gated}
	\begin{algorithmic}[1]
		\State Input: current ordering of vehicles, denoted $(V_{1},V_{2},...,V_{K})$, ordered on basis of $c_{V}$; $V_{last}$, defined as $V_K$ or the last vehicle that crossed the intersection if the ordering is empty; the sets $f_i$ and $t_i$ for $i=1,...,n$ representing the start of platoons and end of platoons at lane $i$; and a to be scheduled vehicle $V_0$ with earliest arrival time at the intersection $a_{V_0}$ in lane $d_{V_0}$.
		\If {$c_{V_{last}}+B<a_{V_0}$}  \Comment{\small $V_0$ is scheduled last}
		\If {$d_{V_0}=d_{V_{last}}$}
		\State Put $c_{V_0}\gets a_{V_0}$.\Comment{\small $V_0$ proceeds without delay}
		\Else
		\State Put $c_{V_0}\gets \max\{a_{V_0},c_{V_{last}}+S\}$.\Comment{\small Check if additional clearance time is needed}
		\EndIf
		\State Add time $c_{V_0}$ to $f_{d_{V_0}}$ and time $c_{V_0}$ to $t_{d_{V_0}}$.\Comment{\small Register $c_{V_0}$ as start of a new platoon}
		\Else
		\If{there is a time in $f_{d_{V_0}}>a_{V_0}$}\Comment{\small $V_0$ is able to join a platoon}
		\State Put $f \gets$ the lowest time in $f_{d_{V_0}}$ such that $f<a_{V_0}$.
        \State Put $t \gets$ the corresponding end of platoon in $t_{d_{V_0}}$.
		\State Put $c_{V_0}\gets t+B$.
		\For {each value $t^*$ in $t_1,...,t_n$ with $t^*> t$}\Comment{\small Update sets $t$ and $f$}
        \State Put $t^*\gets t^*+B$
        \State Put the corresponding start of platoon $f^*\gets f^*+B$.
        \EndFor		
        \For {each vehicle $V$ in the ordering with $c_V>c_{V_0}$}
            \State Put $c_V\gets c_V+B$.\Comment{\small Delay other vehicles}
        \EndFor
		\Else
		\For {$l$ in $(d_{V_0}-1,d_{V_0}-2,...,1,n,n-1,...,d_{V_0}+1)$}
		\If{there is a time in $t_{l}+S>a_{V_0}$}
		\State Find the lowest time $t$ in $t_{l}$ such that $t+S>a_{V_0}$.\Comment{\small $V$ forms a new platoon}
		\State Put $c_{V_0}\gets t+S$.
		\If{there is a time in $f_{l}$ such that $t=f_{l}$}
		\For {each value $t^*$ in $t_1,...,t_n$ with $t^*> t+S$}\Comment{\small Update sets $t$ and $f$}
            \State Put $t^*\gets t^*+2S$
            \State Put the corresponding start of platoon $f^*\gets f^*+2S$.
        \EndFor
		\For {each vehicle $V$ in the ordering with $c_V>t+S$}
            \State Put $c_V\gets c_V+2S$.\Comment{\small Delay other vehicles}
        \EndFor
		\Else \Comment{\small $V_0$ is able to join a platoon}
		\For {each value $t^*$ in $t_1,...,t_n$ with $t^*> t+S$} \Comment{\small Update sets $t$ and $f$}
            \State Put $t^*\gets t^*+S$
            \State Put the corresponding start of platoon $f^*\gets f^*+S$.
        \EndFor
		\For {each vehicle $V$ in the ordering with $c_V>t+S$}
            \State Put $c_V\gets c_V+S$.\Comment{\small Delay other vehicles}
        \EndFor
		\EndIf
		\State Add time $c_{V_0}$ to $f_{d_{V_0}}$ and $c_{V_0}$ to $t_{d_{V_0}}$.
		\State \textbf{break}
		\EndIf
		\EndFor
		\EndIf
		\EndIf
		\If {$C_{V_0}$ is undefined}
		\State Put $c_{V_0}\gets c_{V_K}+B$.
		\State Add time $c_{V_0}$ to $f_{d_{V_0}}$ and time $c_{V_0}$ to $t_{d_{V_0}}$.
		\EndIf
		\State Add vehicle $V_{0}$ to the ordering.
		\State Output: the new ordering $(V_1,V_2,...,V_0,...,V_K)$%, where the vehicles $V$ are ordered on basis of $c_V$.
	\end{algorithmic}
\end{algorithm}

As a reference to algorithms so far established in the literature, we also consider the Batch Algorithm from~\cite{tachet2016revisiting}. For the full description we refer to~\cite[Supplementary~Information,~Section~1.5]{tachet2016revisiting}. The Batch Algorithm has some ingredients of a gated PFA (also in the Batch Algorithm `gates' are put), together with a maximum number of vehicles dealt with in one cycle.

\section{Speed Profile Algorithms} \label{h:arrivalAtI}
Now that we know how to schedule the crossing times of vehicles at the intersection, we turn to the other key ingredient of our model, which is the speed profiling of arriving vehicles. We start with some requirements that the PFAs have to satisfy before we can control the speed of the arriving vehicles in a proper and safe way.
The main condition a PFA has to satisfy is \emph{regularity}.
\begin{definition}[\citealt{miculescu2014polling,miculescu2016polling}] \label{def:regularity}
	A polling policy is \emph{regular} if an arrival in a queue does not change the \emph{order} of service of all currently present vehicles. I.e. the new arrival is inserted somewhere in the order of service of all waiting vehicles.
\end{definition}
A regular polling policy, together with assuming a sufficiently big control region, ensures that the intersection coordination algorithm in~\cite{miculescu2014polling,miculescu2016polling} and the speed profile algorithms that we will introduce are solvable.
%Basically, these two assumptions together ensure the following: at some point in driving towards the intersection you can barely decelerate to zero speed and thereafter speed up to maximum speed again before you cross the intersection. At this time, a vehicle needs to be sure whether or not he can continue. The regularity of the polling policy and a sufficiently big control region are sufficient to achieve this.
These assumptions are necessary with respect to the (possibility of) rescheduling of vehicles. As can be seen in Algorithms~\ref{alg:exhaustive} and ~\ref{alg:gated}, the access time of (some of the) vehicles to the intersection might be increased, upon which trajectories have to be rescheduled. The above assumptions ensure that we can find feasible and safe trajectories for every vehicle, also in case of rescheduling, cf.~\cite{miculescu2014polling,miculescu2016polling}.

Besides these two assumptions on regularity and the size of the control region, we also need to make sure that there are not too many vehicles in the control region at the same time: if there are too many vehicles present in the control region, it might be the case that a newly arriving vehicle cannot decelerate to a complete stop in the distance between entering the control region and the stopping position of its predecessor. This phenomenon is called \emph{overcrowding}, see~\cite{miculescu2016polling}. A way to deal with this issue is proposed as well: we assume that a vehicle that cannot enter the control region safely, does not enter the control region at all.

All PFAs that we discussed are regular in the sense of Definition~\ref{def:regularity}. For the Batch Algorithm of~\cite{tachet2016revisiting} we postulate that this condition is also satisfied.

\subsection{Optimization based Speed Profile Algorithms}

In this subsection, we discuss two algorithms that, satisfying above conditions, result in an efficient use of the intersection which is our main purpose. To this end, we require that vehicles drive at maximum speed while crossing the intersection, so we need to control the speed of arriving vehicles while they are in the control region. A relatively simple optimization algorithm can then be formulated that does the trick, as is shown in~\cite[the MotionSynthesize procedure]{miculescu2014polling,miculescu2016polling}. In order to solve this minimization problem, time is discretized. The MotionSynthesize procedure is then reduced to a linear optimization problem, for which efficient solvers exist.

The optimization procedure has several nice properties, among which is that the algorithm is provably safe. A formal definition of ``safe'' and the required conditions (such as ``no overcrowding'') are given in \cite{miculescu2016polling}, but intuitively it simply means that no collisions will occur in the control region.
Another property is that the distance between vehicle and intersection is minimized across the whole time period a vehicle is in the control region. This is equivalent with the minimization of the area under the distance-time diagram, where the distance is defined as the distance between vehicle and intersection. The physical length of the queue of vehicles is thus also minimized. This is favorable in a network setting, minimizing the amount of spillback to other intersections. Yet, this specific property of minimizing the distance between vehicle and intersection has a high energy consumption and may not be very pleasant for passengers. Below, in Algorithm~\ref{alg:aanrijdenEigen}, we discuss a slightly different formulation of the problem and we minimize the total amount of the absolute value of the \emph{acceleration} instead of the \emph{distance} between vehicle and intersection. Instead of minimizing the area under the distance-time graph, we now minimize the area under the ``absolute value of the acceleration-time'' graph. Note that the scheduled crossing time is set by the PFA. In this section, for consistency with~\cite{miculescu2016polling}, we use the notation $t_f$ to denote scheduled crossing time,  instead of $c_V$. Assuming regularity of the PFA and a sufficiently big control region is not sufficient to ensure a feasible optimization problem as it is for the MotionSynthesize procedure. We formulate a mild additional constraint to guarantee feasibility of the optimization problem, which is that one needs to be sure that when the preceding vehicle is done decelerating, the next vehicle is able to decelerate to that same speed as well before the preceding vehicle is decelerating further (due to rescheduling for example). As will turn out, a vehicle starts decelerating immediately after entering the control region (see e.g. Figure~\ref{f:trajAccExMinL}). As a consequence, if a vehicle is entering the control region, it needs to be sure that it is able to decelerate to the speed of its predecessor while maintaining a certain distance to its predecessor at the same time, showing that we need the additional assumption.

Before we turn to the algorithm, we introduce some notation. Each vehicle has a trajectory that is computed along the lines of the algorithm, given the current time, $t_{0}$, and the scheduled crossing time $t_{f}$. The algorithm will compute $x(t)$, the place of the vehicle at time $t$, for $t_{0}\leq t\leq t_{f}$, the speed $v(t)$ at time $t$ and the acceleration $a(t)$ at time $t$. Furthermore, $y(t)$ denotes the trajectory of the predecessor (if any) for $t_{0,y}\leq t\leq t_{f,y}$; $t_{f,y}$ denotes the final crossing time of the predecessor of the vehicle we are currently planning; $l$ denotes the minimal distance between the front part of two successive vehicles; $a_m$ denotes the maximum acceleration; $-a_m$ denotes the maximum deceleration; and $v_m$ denotes the maximum speed. The initial conditions, i.e. the location and speed at the start of the trajectory of the vehicle, are given by $x(t_{0})=x_0$ and $v(t_{0})=v_0$. To put the location $x(t)$ into perspective, we measure $x(t)$ as the (negative) distance between the vehicle and the start of the conflict area of the intersection, i.e. $x(t_0)=x_0=-X$ and $x(t_f)=0$, when the vehicle enters the control region at a distance $X$ from the intersection.
Algorithm~\ref{alg:aanrijdenEigen} can be discretized in order to obtain a linear optimization problem, just as the MotionSynthesize procedure.

\begin{algorithm}
\caption{MotionSynthesize procedure with a minimal acceleration}\label{alg:aanrijdenEigen}
	\begin{algorithmic}[1]
		\State Input: $x_0$, $v_{0}$, $t_{0}$,
                $t_{f}$,$t_{f,y}$, $y$.
		\State Compute \begin{align}
		\textrm{MotionSynthesizeAcc}(x_0,v_0,t_{0},t_{f},t_{f,y},y) &:= \underset{x:[t_{0},t_{f}]\to \mathbb{R}}\argmin \int_{t_{0}}^{t_{f}} |a(t)|dt\nonumber \\\nonumber
		& \textrm{subject to } \\\nonumber
		&x^{\prime\prime}(t)=a(t) \text{, for all } t\in[t_{0},t_{f}];\\\nonumber
		&0\leq x^\prime(t)\leq v_m \text{, for all } t\in[t_{0},t_{f}];\\\nonumber
		&|a(t)| \leq a_m\text{, for all } t\in[t_{0},t_{f}];\\\nonumber
		&|x(t)-y(t)|\geq l \text{, for all } t\in[t_{0},t_{f,y}];\\\nonumber
		&x(t_{0}) =x_{0} \textrm{; }x^{\prime}(t_{0})=v_0 ;\\\nonumber
		&x(t_{f})=0 \textrm{; } x^\prime(t_{f})=v_m.
		\end{align}
		\State Output: $x(t)$.
	\end{algorithmic}
\end{algorithm}

Algorithm~\ref{alg:aanrijdenEigen} is solvable under the set of conditions formulated above, i.e. regularity of the PFA, a sufficiently big control region and the assumption on decelerating of a predecessor of a vehicle.
The main difference is that instead of minimizing the distance from vehicle to intersection, we minimize the (absolute value of the) acceleration applied by the vehicle while being in the control region. This obviously has consequences for the amount of energy consumption (it will be lower than in the MotionSynthesize procedure). Disadvantages include that the physical length of the queue grows and that vehicles cannot enter the control region as close to each other (as vehicles slow down immediately when entering the control region). %A full proof of the statement that Algorithm 3 is solvable under the given conditions is not given here.

In the next subsection we present closed-form solutions to the MotionSynthesize procedure and Algorithm~\ref{alg:aanrijdenEigen}, similar in spirit as the results in e.g.~\cite{lawitzky2013energy} and~\cite{dib2014optimal}. So instead of the need to solve a linear optimization problem each time, we have a simple set of calculations that we can perform to find the trajectory of a vehicle, which is optimal with respect to minimizing the distance or acceleration.
These closed-form expressions immediately show why Algorithm~3 is solvable. In Remark \ref{rem:solv} we return to this topic.

\subsection{Closed-form Speed Profile Algorithms}

We start with two important observations that form the basis for our closed-form SPAs:
\begin{enumerate}
\item the optimization problem formulated in the MotionSynthesize procedure always leads to piece-wise constant acceleration;
\item if all vehicles decelerate (and possibly stop) at most \emph{once}, at most four changes in the acceleration occur.
\end{enumerate}
%Below we will see why.

These observations imply that if we can find the points at which the acceleration changes, we are able to determine the trajectory analytically and in closed-form. We give some intuition behind the main ideas of Algorithm~\ref{alg:aanrijdenAna}. We have to plan the trajectory from $t_0$, the current time, until $t_f$, the crossing time set by the PFA. It is sufficient to give the acceleration for any time $t\in[t_0, t_f]$. Indeed, it is true that the exhaustive and gated algorithms have the desirable property that vehicles need to decelerate at most once. From the polling literature we know that exhaustive service ensures that customers will always be served before the end of the cycle in which they arrive. With gated service, customers will always be served in the next cycle. Translated to our traffic model, this means that no vehicle will ever need to stop more than once. As a consequence, the acceleration is piece-wise constant and changes at most four times.
We shortly describe those five parts of the arriving trajectory.

\begin{itemize}
	\item No acceleration or deceleration from $t_0$ until $t_{\textit{dec}}$;
	\item Deceleration at maximum rate from $t_{\textit{dec}}$ until $t_{\textit{stop}}$;
	\item A stop from $t_{\textit{stop}}$ until $t_{\textit{acc}}$;
	\item Acceleration at maximum rate from $t_{\textit{acc}}$ until $t_{\textit{full}}$;
	\item No acceleration or deceleration from $t_{\textit{full}}$ until $t_{f}$.
\end{itemize}
All that remains is that we have to find $t_{\textit{dec}}$, $t_{\textit{stop}}$, $t_{\textit{acc}}$ and $t_{\textit{full}}$ in such a way that we minimize the average distance between vehicle and intersection.
%The four times are found using the following observations. Entering the control region, we continue as long as possible at maximum speed, decreasing the distance between intersection and vehicle as quickly as possible (this corresponds to time $t_{\textit{dec}}$). After some time we know that we have to decelerate, because if we do not, we either are too early at the intersection or do not drive at maximum speed. This implies that the remainder of the trajectory is fixed: we decelerate at maximum rate (until $t_{\textit{stop}}$), possibly stop for some time (until $t_{\textit{acc}}$), and accelerate at maximum rate (until $t_{\textit{full}}$). Then the vehicle might drive at full speed for some time until $t_f$.
This is done in the closed-form solution of the MotionSynthesize procedure (Algorithm~\ref{alg:aanrijdenAna}), where we assume that $t_{0}=0$ to ease the notation and that $v_0=v_m$. We can allow for general $v_0$, but we show later that this would always result in a sub-optimal trajectory. The input consists of the (negative) distance between vehicle and intersection at $t=0$, again denoted by $x_0$, the scheduled crossing time of the vehicle, $t_f$, and the trajectory of the predecessor of the vehicle for which we are currently planning the trajectory, $y$, and its crossing time $t_{f,y}$ .
We  prove that the MotionSynthesize procedure and Algorithm~\ref{alg:aanrijdenAna} are equivalent, which is the subject of the next lemma.

\begin{algorithm}
	\caption{closed-form solution to the MotionSynthesize procedure.}\label{alg:aanrijdenAna}
	\begin{algorithmic}[1]
		\State Input: $x_0$, $t_{f}$, $t_{f,y}$, and $y$.
%		\State Define the arrival time of the vehicle associated with $y$ as $t_{f,y}$.
		\If {$t_f-t_{f,y}=B$} \label{algAnl:3}
		\State Consider trajectory $y$ and determine the time at which the vehicle continues at full speed. Call this time $t_{\textit{full}}$.\label{algAnl:4}
		\Else \label{alglAn:5}
		\State Put $t_{\textit{full}}\gets t_{f}$.
		\EndIf \label{alglAn:7}
		\State  Put
		\begin{equation}\label{eq:Lp} L \gets  v_m\left(t_f-\frac{v_m}{a_m}\right). \end{equation} \label{alg:anaLoc9}\Comment{$L$ represents the distance covered if a vehicle stops for 0 seconds}
		\If{$L\geq |x_0|$ } \Comment{The vehicle has to stop}
		\State Put $t_{\textit{acc}}\gets t_{\textit{full}}-v_m/a_m$.\label{algAnl:10}
		\State Put $t_{\textit{stop}}\gets t_{\textit{acc}}-(t_{f}-v_m/a_m-|x_0|/v_m)$.
		\State Put $t_{\textit{dec}}\gets t_{\textit{stop}}-v_m/a_m$. \label{algAnl:12}
		\Else \Comment{The vehicle does not have to stop}
		\State Define
		\begin{equation}\label{eq:t}
		\tilde{t} \gets  \sqrt{\frac{t_{f}v_m-|x_0|}{a_m}}.
		\end{equation}\Comment{$\tilde{t}$ is the deceleration time}
		\State Put $t_{\textit{acc}}\gets t_{\textit{full}}-\tilde{t}$.\label{algAnl:13}
		\State Put $t_{\textit{stop}} \gets  t_{\textit{acc}}$.
		\State Put $t_{\textit{dec}}\gets t_{\textit{acc}}-\tilde{t}$.\label{algAnl:15}
		\EndIf
		\State Then \begin{equation} \label{eq:acce}
		a(t) =  x^{\prime\prime}(t) \gets \begin{cases}
		0 &\textrm { if } 0 \leq t < t_{\textit{dec}}, \\
		-a_m & \textrm{ if } t_{\textit{dec}} \leq t < t_{\textit{stop}}, \\
		0 & \textrm{ if }t_{\textit{stop}} \leq t < t_{\textit{acc}}, \\
		a_m & \textrm { if } t_{\textit{acc}} \leq t < t_{\textit{full}}, \\
		0 & \textrm { if } t_{\textit{full}} \leq t < t_f.
		\end{cases}
		\end{equation}
		\State Knowing $a(t)$, we can compute $x(t)$ by integrating twice and using the conditions $x(0)=x_0$ and the velocity at time 0 being $v_m$.
		\State Output: $x(t)$.
	\end{algorithmic}
\end{algorithm}

\begin{lemma}\label{l:3=5}
	The MotionSynthesize procedure and Algorithm~\ref{alg:aanrijdenAna} are equivalent in the sense that both minimize the distance between vehicle and intersection across the time period $t_0$ to $t_f$.
\end{lemma}

\begin{proof}
We split the proof in two parts. First we prove that the times $t_{\textit{dec}}$, $t_{\textit{stop}}$, $t_{\textit{acc}}$ and $t_{\textit{full}}$ in Algorithm~\ref{alg:aanrijdenAna} indeed result in the trajectory having the minimal area under the distance-time graph, \emph{assuming that} the optimal trajectory contains at most one period of deceleration. Then we prove that the obtained \emph{form} of the trajectory, with at most one period of deceleration, is indeed optimal.

\begin{proofpart}
As indicated before, for now, we only consider trajectories that contain at most \emph{one} period of deceleration. We allow that $v_0<v_m$ (but we will show now that that is suboptimal), but we do require that $v(t_{\textit{full}})=v(t_f)=v_m$. We distinguish between the case where a vehicle comes to a full stop and the case where it does not.

\paragraph*{Full stop. }
First we consider the case where the vehicle (denoted by $V$) comes to a full stop, from $t=t_{\textit{stop}}$ to $t=t_{\textit{acc}}$. This class of trajectories is visualised as the black line in Figure~\ref{fig:optimalTrajectories}. It turns out that this curve is completely characterised by two parameters, which we choose to be the initial speed $v_0$ and the moment when we start driving at full speed again, $t_{\textit{full}}$.

\begin{figure}[!ht]
\begin{center}
\textsl{2}\includegraphics[width=0.7\textwidth]{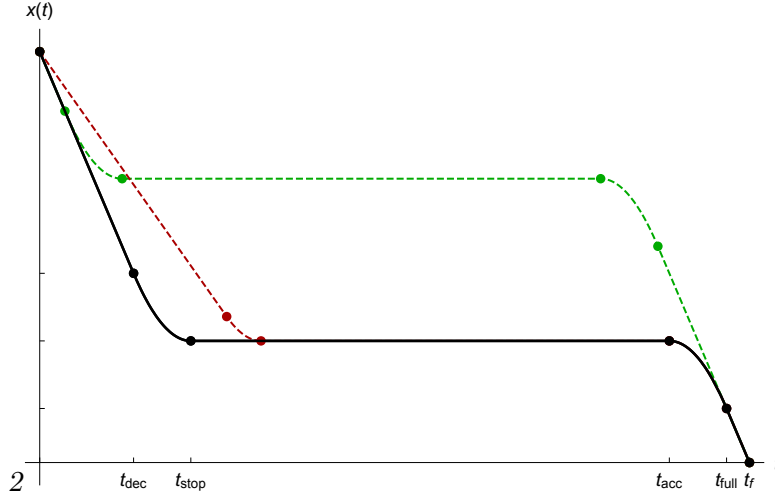}
\end{center}
\caption{Three sample trajectories with one full stop. The optimal trajectory is plotted in black. The dashed green trajectory has a smaller value of $t_{\textit{full}}$ compared to the optimal trajectory, whereas the dashed red trajectory has a smaller value of $v_0$.}
\label{fig:optimalTrajectories}
\end{figure}

The optimization criterion in the MotionSynthesize algorithm is to minimize the area below the graph $|x(t)|$ for $0\leq t \leq t_f$. This is equivalent to minimizing the average distance to the intersection. First we will give an intuitive explanation as to why it makes sense to continue at full speed as long as possible. In Figure~\ref{fig:optimalTrajectories} we have plotted two alternative trajectories to show that they result in a larger average distance to the intersection. The red dashed trajectory is equivalent to the optimal trajectory, but with a lower starting speed $(v_0<v_m)$. By starting at a lower speed, while fixing $t_{\textit{full}}$, we have to continue longer at this lower speed before we come to a complete stop. This means that $t_{\textit{dec}}$ and $t_{\textit{stop}}$ increase, which immediately increases the area below the graph. Another alternative is the dashed green trajectory, which starts at full speed, but has a lower value for $t_{\textit{full}}$. Note that $t_{\textit{full}}$ is restricted by $V$'s predecessor. Without predecessor, it is optimal to take $t_{\textit{full}}=t_f$, but if there is a predecessor (which apparently is the case for the black trajectory in Figure~\ref{fig:optimalTrajectories}), it is optimal to let both vehicles have the same $t_{\textit{full}}$. This is the only way to ensure that both vehicles cross the intersection at full speed, with minimum distance between them. Taking a smaller value of $t_{\textit{full}}$, as in the green trajectory, means that $V$ comes to a stop further from the intersection, which significantly increases the average distance.

These arguments provide an intuitive explanation, but we will formalize this now by explicitly computing the area below $|x(t)|$ for our closed-form trajectories. First we give the closed-form expression for $x(t)$, by considering the five sub-areas separately, and using the fact that $x(t)$ is linear when the speed is constant and quadratic while decelerating/accelerating. Equation \eqref{eqn:xt} is easiest to understand when starting at $t=t_f$ and constructing the trajectory backwards to $t=0$, and using these
auxiliary results:
\[
t_{\textit{stop}}-t_{\textit{dec}}=\frac{v_0}{a_m}, t_{\textit{full}}-t_{\textit{acc}}=\frac{v_m}{a_m},
x(t_{\textit{stop}})-x(t_{\textit{dec}})=\frac{v_0^2}{2a_m}, x(t_{\textit{full}})-x(t_{\textit{acc}})=\frac{v_m^2}{2a_m}.
\]
We obtain:
\begin{equation}
x(t)=\begin{cases}
(t-t_f)v_m & \quad \text{ for }t_{\textit{full}}\leq t\leq t_f,\\
(t_{\textit{full}}-t_f)v_m - \frac{v_m^2}{2a_m} + \frac{a_m}{2}(t-t_{\textit{acc}})^2 & \quad \text{ for }t_{\textit{acc}}\leq t\leq t_{\textit{full}},\\
(t_{\textit{full}}-t_f)v_m - \frac{v_m^2}{2a_m} &\quad \text{ for }t_{\textit{stop}}\leq t\leq t_{\textit{acc}},\\
(t_{\textit{full}}-t_f)v_m - \frac{v_m^2}{2a_m} - \frac{a_m}{2}(t-t_{\textit{stop}})^2   &\quad \text{ for }t_{\textit{dec}}\leq t\leq t_{\textit{stop}},\\
x_0 + v_0 t &\quad \text{ for }0\leq t\leq t_{\textit{dec}}.\\
\end{cases}
\label{eqn:xt}
\end{equation}
Note that $t_{\textit{dec}}$ follows from continuity of $x(t)$:
\[
%x_0-v t_{\textit{dec}} = (t_f-t_{\textit{full}})v_m + \frac{v_m^2}{2a_m} + \frac{a_m}{2}\left(\frac{v_0}{a_m}\right)^2.
t_{\textit{dec}}=\frac{1}{v_0}\left(|x_0|-(t_f-t_{\textit{full}})v_m-\frac{v_0^2+v_m^2}{2a_m}\right).
\]
The area below the trajectory, $\displaystyle\mathcal{A}_v:= \int_0^{t_f}|x(t)|\text{d}t$, is equal to:
\begin{align}
\mathcal{A}_v =\,&
\frac{t_{\textit{dec}}}{2}\left(x(t_{\textit{dec}})-x_0+\frac{v_0^2}{a_m}\right)+ \frac{v_0^3}{6a_m^2}\nonumber\\ &+t_{\textit{full}}\left((t_f-t_{\textit{full}})v_m + \frac{v_m^2}{2a_m}\right)-\frac{v_m^3}{6a_m^2}+ \frac{v_m}{2} (t_f-t_{\textit{full}})^2\nonumber\\
=\,&\frac{v_0^4+3(v_m^2+2a_m((t_f-t_{\textit{full}})v_m+x_0))^2}{24a_m^2v_0}\nonumber\\
&+\frac{v_m}{2}\left(t_f^2-t_{\textit{full}}^2+t_{\textit{full}}\frac{v_m}{a_m}\right) -\frac{v_m^3}{6a_m^2}.
\label{eqn:xtArea}
\end{align}
We now exploit that only the first part of the expression for $\mathcal{A}_v$ depends on the initial speed $v_0$, as observed before. By taking the derivative with respect to $v_0$ and using $v_0\leq v_m$ it follows that $\mathcal{A}_v$ is decreasing in $v_0$, under the following condition:
\[
(t_f-t_{\textit{full}})v_m+2\frac{v_m^2}{2a_m} \leq |x_0|.
\]
This is exactly the ``no overcrowding'' assumption discussed earlier, which now gets quantified: a vehicle entering the control region at full speed should have enough space to come to a full stop and accelerate again in order to reach full speed at time $t_{\textit{full}}$. This proves that the initial speed should be taken as large as possible, i.e. $v_0=v_m$.

Now that we have established that we should choose $v_0=v_m$, we assume this equality from now on and denote the area as $\mathcal{A}$ (to distinguish it from $\mathcal{A}_v$). This significantly simplifies the expression, which now becomes
\begin{align}
\mathcal{A} =\,&
(v_m t_f+x_0)\left(t_f-t_{\textit{full}}+\frac{v_m}{2a_m}\right)+\frac{x_0^2}{2v_m}.
\label{eqn:xtArea2}
\end{align}
It is readily seen that the area $\mathcal{A}$ is now \emph{linearly decreasing} in $t_{\textit{full}}$, which immediately proves that we should take $t_{\textit{full}}$ as large as possible to minimize $\mathcal{A}$. Exactly how large $t_{\textit{full}}$ is allowed to be, depends on the predecessor.

\paragraph*{No full stop. }
	
We now briefly consider the case where $V$ does not come to a full stop. The analysis is quite similar, so we will mainly focus on the differences. The first difference is that $t_{stop}$ is removed from the trajectory. Instead, we now have that the speed at $t=t_\textit{acc}$ is greater than zero. Note that this speed, which we denote by $v_1$, is less than or equal to $v_0$, because $V$ decelerates between $t_\textit{dec}$ and $t_\textit{acc}$. The trajectory $x(t)$ now consists of at most four parts,  given by:
\begin{equation}
x(t)=\begin{cases}
(t-t_f)v_m & \quad \text{ for }t_{\textit{full}}\leq t\leq t_f,\\
(t-t_f)v_m + \frac{a_m}{2}(t-t_{\textit{full}})^2 & \quad \text{ for }t_{\textit{acc}}\leq t\leq t_{\textit{full}},\\
x_0 + v_0 t - \frac{a_m}{2}(t-t_{\textit{dec}})^2 & \quad \text{ for }t_{\textit{dec}}\leq t\leq t_{\textit{acc}},\\
x_0 + v_0 t &\quad \text{ for }0\leq t\leq t_{\textit{dec}}.\\
\end{cases}
\label{eqn:xtNoStop}
\end{equation}
We can eliminate the unknowns by using the relations
\[
t_{\textit{acc}}-t_{\textit{dec}}=\frac{v_0-v_1}{a_m}, t_{\textit{full}}-t_{\textit{acc}}=\frac{v_m-v_1}{a_m}.
\]
The requirement that $x(t)$ is continuous in $t_\textit{acc}$ leads to the last equation that can be solved to obtain $t_{\textit{acc}}$. The area below $|x(t)|$ can now be computed:
\begin{align}
\mathcal{A}_v =\,&\frac{v_m}{2}(t_f-t_{\textit{acc}})^2+\frac{(v_0-v_1)^3-(v_m-v_1)^3}{6a_m^2}-x_0t_{\textit{acc}}-\frac{v_0}{2} t_{\textit{acc}}^2 .
\label{eqn:xtAreaNoStop}
\end{align}
Eliminating $t_{\textit{acc}}$ and differentiating with respect to $v_1$ immediately shows that $\mathcal{A}_v$ is decreasing in $v_1$. Since we are trying to minimize $\mathcal{A}_v$, we should take $v_1$ as large as possible, i.e. $v_1=v_0$. After this substitution, all expressions simplify and it can again be shown that the derivative of $\mathcal{A}$ with respect to $v_0$ is always less than or equal to zero, where equality is only reached when $t_\textit{acc}=0$ and there is no other option for $V$ than to accelerate immediately. This means that we should take $v_0$ as large as possible, which again implies that we should take $t_\textit{full}$ as large as possible.

It should be noted that the case $v_0=v_m$ needs to be considered separately, because if the conditions allow a maximal initial speed, $v_1$ is completely fixed:
\[
v_1=v_m-\sqrt{a_m(t_fv_m-|x_0|)}.
\]
This means that $t_\textit{full}$ does not follow from $v_0$, but it can be chosen arbitrarily (between the minimum and maximum allowed values). In this case we have
\[
t_\textit{acc}=t_{\textit{full}}-\frac{v_m-v_1}{a_m} = t_{\textit{full}}-\frac{v_m-\big(v_m-\sqrt{a_m(t_fv_m-|x_0|)}\big)}{a_m} = t_{\textit{full}}-\tilde{t},
\]
with $\tilde{t}$ as defined in \eqref{eq:t}.

\paragraph*{Implementation. } Algorithm \ref{alg:aanrijdenAna} is an implementation of the optimal trajectory for the general case. The formulation of the algorithm is slightly different, because we are using the results that $v_0$ and $t_\textit{full}$ should be as large as possible. As argued above, an upper bound to the time $t_{\textit{full}}$ is determined by the trajectory $y$ of the predecessor of $V$, and is fixed. If the crossing times differ a time $B$, then the time at which the predecessor starts driving at full speed, $t_{f,y}$, should equal $t_{\textit{full}}$ (because we want to take it as large as possible), and otherwise it is simply $t_f$, which is the way we choose $t_{\textit{full}}$ in lines~\ref{algAnl:3}-\ref{alglAn:7}.

The astute reader will also notice that we do not provide an explicit expression for $x(t)$ in Algorithm \ref{alg:aanrijdenAna}. Instead, we provide its second derivative, $a(t)$, and the boundary conditions. This has the advantage that we have one formulation that is valid for both cases (full stop and no full stop). One can easily verify that \eqref{eqn:xt} (full stop) and  \eqref{eqn:xtNoStop} (no full stop) both reduce to \eqref{eq:acce} after differentiating twice, and that $t_\textit{dec}$, $t_\textit{stop}$, $t_\textit{acc}$ and $t_\textit{full}$ as computed in Algorithm \ref{alg:aanrijdenAna} correspond to the values discussed in the first part of the proof. Note that we choose $t_{\textit{stop}}=t_{\textit{acc}}$ in case of no full stop.
	
Then combining the defined times, we obtain \eqref{eq:acce}, which minimizes the area under the distance-time graph. This is exactly the same criterion as we optimize for in the MotionSynthesize procedure. The only thing left to show, is that all other trajectories satisfying the required constraints regarding maximum speed and acceleration, have a larger average distance to the intersection than the one we obtain.
\end{proofpart}

%We split the proof in two parts. First we prove that the obtained \emph{form} of the trajectory is optimal (i.e. the five parts of the trajectory) and then we prove that the times $t_{\textit{dec}}$, $t_{\textit{stop}}$, $t_{\textit{acc}}$ and $t_{\textit{full}}$ in Algorithm~\ref{alg:aanrijdenAna} indeed result in the trajectory having the minimal area under the distance-time graph.

\begin{proofpart}
This part is significantly shorter, proving that the obtained trajectory is really optimal with respect to the criterion of smallest average distance to the intersection. We remind the reader that we explicitly exploit the property of the polling-based PFAs that each vehicle needs to decelerate (and possibly stop) at most once. Intuitively, the optimality is quite apparent: in order to minimize the average distance to the intersection, a vehicle entering the control region needs to drive at full speed as long as possible. Assume that $x(t)$ is a trajectory defined by \eqref{eqn:xt} with $v_0 = v_m$ and $t_{\textit{full}}$ as large as possible. We now consider an alternative trajectory $\tilde{x}(t)\neq x(t)$. We compare $x(t)$ with $\tilde{x}(t)$ on the five parts of the trajectory.
\begin{itemize}
\item For $0\leq t\leq t_{\textit{dec}}$ it is completely obvious that $|\tilde{x}(t)|\geq |x(t)|$, because $\tilde{x}(0)=x(0) = x_0$ and $\tilde{x}'(t)\leq x'(t) = v_m$ for $0\leq t\leq t_{\textit{dec}}$.
\item We now turn to the \emph{last} part of the trajectory. For $t_{\textit{full}}\leq t\leq t_{f}$ it is we have $\tilde{x}(t)= x(t)$ because $t_{\textit{full}}$ was defined as the largest possible value for $t$ where $V$ should start driving at full speed.
\item Looking at the part before this one, $t_{\textit{acc}}\leq t\leq t_{\textit{full}}$, we see that $|\tilde{x}(t)|\geq |x(t)|$ because $\tilde{x}'(t_{\textit{full}})=x'(t_{\textit{full}})=v_m$ and $\tilde{x}''(t)\leq x''(t)=a_m$.
\item The period $t_{\textit{stop}}\leq t\leq t_{\textit{acc}}$ is also trivial, because $\tilde{v}(t)\geq v(t)=0$ here, meaning that $|\tilde{x}(t)|\geq |x(t)|$.
\item This leaves us with the last part, which is the second period $t_{\textit{dec}}\leq t\leq t_{\textit{stop}}$. We have already established that $|\tilde{x}(t_{\textit{dec}})|\geq |x(t_{\textit{dec}})|$ and $|\tilde{x}(t_{\textit{stop}})|\geq |x(t_{\textit{stop}})|$. Since $\tilde{x}'(t_{\textit{dec}})\leq x'(t_{\textit{dec}})=v_m$ and $\tilde{x}''(t)\leq x''(t)=a_m$, it also follows that $|\tilde{x}(t)|\geq |x(t)|$ in this area.
\end{itemize}
The conclusion is that for all $t\in[0,t_f]$ we have $|\tilde{x}(t)|\geq |x(t)|$, which implies that
\[
\int_0^{t_f}|\tilde{x}(t)|\text{d}t \geq \int_0^{t_f}|x(t)|\text{d}t.
\]
This proves that the path $x(t)$ is optimal with respect to the criterion of the MotionSynthesize procedure. Since it has also been proven in \cite{miculescu2014polling} that the MotionSynthesize algorithm yields an optimal path, both algorithms must return the same path.
\end{proofpart}
\end{proof}

\begin{remark}
Although the exhaustive and gated PFA ensure that there is at most \emph{one} period of deceleration, possibly a stop, and acceleration, for other disciplines, like the Batch Algorithm or the $k$-limited discipline (which is also based on polling models), this might not be the case, and the period from $t_0$ until $t_f$ might have to be split in more than five different periods. A similar type of speed profile algorithm is still possible, but is more involved and therefore omitted in interest of space and clarity of the algorithm and argumentation.
\end{remark}

So, Algorithm~\ref{alg:aanrijdenAna} has the same desirable properties as the MotionSynthesize procedure, but is computationally much less expensive and also provides intuition on the shape of the trajectories. A visualization of such trajectories can be found in Figure~\ref{f:trajAccExMinL}.

\null

\noindent We can also formulate such an alternative for Algorithm~\ref{alg:aanrijdenEigen}, where we, again, put $t_0=0$ to ease the notation. We allow for general $v_{0}$ now. In fact, this is essential to this algorithm, because a vehicle might start decelerating immediately upon arrival in the SPA part of the control region. We assume that a following vehicle has decelerated accordingly, if necessary, in the PFA part of the control region. In practice, either vehicle-to-vehicle (V2V) or controller-to-vehicle (V2I) communication might be used to ensure this speed adjustment.
The general structure of Algorithm~\ref{alg:aanrijdenEigen} is similar to Algorithm~\ref{alg:aanrijdenAna}.
Also in this case, the acceleration is piece-wise constant, yet there are at most three changes in the acceleration. We shortly describe those four parts of the arriving trajectory.
\begin{itemize}
	\item Deceleration at maximum rate from $t_{0}$ until $t_{cruise}$;
	\item No acceleration or deceleration from $t_{cruise}$ until $t_{\textit{acc}}$;
	\item Acceleration at maximum rate from $t_{\textit{acc}}$ until $t_{\textit{full}}$;
	\item No acceleration or deceleration from $t_{\textit{full}}$ until $t_{f}$.
\end{itemize}
 This is also visible in Figure~\ref{f:trajAccExMinL}.
Note that we start decelerating as soon as possible, because we want to cruise at a relatively low speed. If we would not cruise at a low speed, then we would have to decelerate more (as we covered a longer distance at a high speed). So we decelerate maximally for some time, continue at a constant speed for some time and then accelerate maximally (taking advantage of the lower cruising speed as long as possible). The resulting algorithm is formulated in Algorithm~\ref{alg:aanrijdenEigenAna} and equivalence with Algorithm~\ref{alg:aanrijdenEigen} is proven thereafter.

\begin{algorithm}[!ht]
	\caption{closed-form solution to Algorithm~\ref{alg:aanrijdenEigen}.}\label{alg:aanrijdenEigenAna}
	\begin{algorithmic}[1]
		\State Input: $x_0$, $v_0$, $t_{f}$, $t_{f,y}$, and $y$.
%		\State Define the arrival time of the vehicle associated with $y$ as $t_{f,y}$.
		\If {$t_f-t_{f,y}=B$}
		\State Consider trajectory $y$ and determine the time at which the vehicle continues at full speed. Call this time $t_{\textit{full}}$.
		\Else
		\State Put $t_{\textit{full}}\gets t_{f}$.
		\EndIf
		\State  Put
		\begin{equation}\label{eq:cruise}
		\begin{aligned}
		t_{1} &\gets  \frac{a_mt_f+v_{0}-v_m}{2a_m}-\\&\frac{\sqrt{4a_m |x_0| +(a_m t_{f}-v_{0})^2-2(a_mt_fv_m+v_0^2)-4a_m(t_f-t_{\textit{full}})v_m+2v_{0}v_m-v_{m}^2}}{2a_m}
		\end{aligned}
		\end{equation}
		\State Put
		\begin{equation}\label{eq:tac}
		\begin{aligned}
		t_{2} &\gets  \frac{a_mt_f+v_{0}-v_m}{2a_m}+\\&\frac{\sqrt{4a_m |x_0| +(a_mt_{f}-v_{0})^2-2(a_mt_fv_m+v_0^2)-4a_m(t_f-t_{\textit{full}})v_m+2v_{0}v_m-v_{m}^2}}{2a_m}		\end{aligned}
		\end{equation}
		\State Put $t_{cruise} = t_1$ and $t_{\textit{acc}}=t_2$.
		\State Then,
		\begin{equation}\label{eq:acceOwn}
		a(t) =  x^{\prime\prime}(t) \gets \begin{cases}
		-a_m &\textrm { if } 0 \leq t < t_{cruise}, \\
		0 & \textrm{ if } t_{cruise} \leq t < t_{\textit{acc}}, \\
		a_m & \textrm { if } t_{\textit{acc}} \leq t < t_{\textit{full}}, \\
		0 & \textrm { if } t_{\textit{full}} \leq t < t_f.
		\end{cases}
		\end{equation}
		\State Knowing $a(t)$, we can compute $x(t)$ by integrating twice and using the conditions $x(0)=x_0$ and $v(0)=v_0$.
		\State Output: $x(t)$.
	\end{algorithmic}
\end{algorithm}

%Then we have the following lemma.

\begin{lemma} \label{l:4=6}
Algorithm~\ref{alg:aanrijdenEigen} and Algorithm~\ref{alg:aanrijdenEigenAna} are equivalent in the sense that both minimize the absolute value of the applied acceleration across the time period $t_0$ to $t_f$.
\end{lemma}

\begin{proof}
	We again split the proof in two parts, but now we first prove optimality of the form of the trajectory and then we check the computation of $t_{cruise}$, $t_{\textit{acc}}$ and $t_{\textit{full}}$ in Algorithm~\ref{alg:aanrijdenEigenAna}.
	
	\begin{proofpart}
		The optimal trajectory consists of at most four parts. The last part, from $t_{\textit{full}}$ until $t_f$ is determined in the same way as shown in the proof of Lemma~\ref{l:3=5}.
		
		The first three parts of the trajectory are split in the following way: decelerating (until $t_{cruise}$), cruising at a fixed speed (until $t_{\textit{acc}}$) and accelerating (until $t_{\textit{full}}$), where the first and last period may have zero length. We want to minimize the area under the absolute value of the acceleration-time graph. We decelerate as early as possible and accelerate as late as possible, and both at the maximum rate. If we would not do one of these three things, it means that we would have to decelerate more as we drive at a high speed longer (and as e.g. the average speed is fixed, namely $x_{0}/t_f$, we would have to decelerate more to a lower speed). So, indeed the first three parts of a trajectory consist of decelerating at maximum rate, then cruising at a fixed (and relatively low) speed and then accelerating at maximum rate.
	\end{proofpart}
	\begin{proofpart}
		As argued in the proof of Lemma~\ref{l:3=5}, the time $t_{\textit{full}}$ is determined by the trajectory $y$ of the predecessor of $V$ and is fixed. So $t_{\textit{full}}$ is chosen as in lines~\ref{algAnl:3}-\ref{alglAn:7}.
		
		Knowing this, we can compute the remainder of the trajectory. We can compute the traversed distance if we immediately decelerate for a time $t$ and accelerate as late as possible for a time $t+v_m/a_m-v_{0}/a_m$ (because it might be that $v_0\neq v_m$), which is
		 \begin{equation}
		\begin{aligned} &v_{0}t-\frac{1}{2}a_mt^2+\left(v_m-a_m(t+\frac{v_m}{a_m}-\frac{v_{0}}{a_m})\right)\left(t+\frac{v_m}{a_m}-\frac{v_{0}}{a_m}\right)+(t_{f}-t_{\textit{full}})v_m +\\&\left(v_m-a_m(t+\frac{v_m}{a_m}-\frac{v_{0}}{a_m})\right)
		\left(t_f-2t-\frac{v_m}{a_m}+\frac{v_{0}}{a_m}\right)+\frac{1}{2}a_m\left(t+\frac{v_m}{a_m}-\frac{v_{0}}{a_m}\right)^2. \label{eq:findT}\end{aligned}\end{equation}
		Equating \eqref{eq:findT} with $|x_0|$ and solving for $t$, results in two positive values. The smaller one is given as $t_1$ in \eqref{eq:cruise} and the larger one as $t_2$ in \eqref{eq:tac}. So we can put $t_{cruise}=t_1$ and $t_{\textit{acc}} =t_2$.
		
		Then combining the defined times, we obtain \eqref{eq:acceOwn}. With this choice of times, we see that we minimize the area under the absolute value of the acceleration-time graph. This is exactly the same criterion as we optimize for in Algorithm~\ref{alg:aanrijdenEigen}, so the two algorithms yield the same trajectory.
	\end{proofpart}
\end{proof}

\begin{remark}\label{rem:solv}
Algorithms \ref{alg:aanrijdenEigen}  and \ref{alg:aanrijdenEigenAna}  are solvable, if the PFA is regular, the control region is sufficiently big and the cars are sufficiently far apart from each other when entering the control region (as mentioned before). The regularity of the PFA ensures that the vehicles keep driving behind each other (and do not have to overtake). Our closed-form expressions in Algorithm \ref{alg:aanrijdenEigenAna} provide immediate quantitative insight in the conditions required for solvability.
In this case, lines 2 to 6 are sufficient to determine the influence of the predecessor of the vehicle that we are currently planning. The sufficiently big control region ensures that proper $t_{\textit{full}}$, $t_1$ and $t_2$ can be found, in such a way that vehicles do not collide, which is also the case for the lost condition on the distance between cars when they enter the control region. A full proof would be similar to the proof of Lemma (IV.4) in \cite{miculescu2016polling} and would follow along the same lines.
\end{remark}

 A visualization of trajectories generated by Algorithms~\ref{alg:aanrijdenAna} and~\ref{alg:aanrijdenEigenAna} is depicted in Figure~\ref{f:trajAccExMinL}.

\begin{figure}
	\centering	\includegraphics[width=0.7\linewidth]{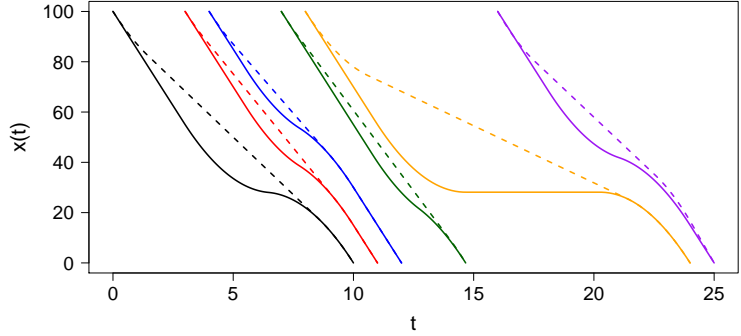}
	\caption{Algorithm~\ref{alg:aanrijdenAna} (solid lines) and Algorithm~\ref{alg:aanrijdenEigenAna} (dashed lines) for several vehicles with $t$ (s) on the horizontal axis and $|x(t)|$ (m) on the vertical axis for several vehicles.}
	\label{f:trajAccExMinL}
\end{figure}
\section{Performance Analysis}\label{h:analysis}

Having covered the two main ingredients of the model, we turn to the performance analysis. The two measures that we consider are mean delay and fairness. In order to obtain results on mean delay and fairness, we first establish a link between the model we described so far and polling models.

\subsection{Polling Model} \label{h:polling}

Polling models are well-studied mathematical objects representing queueing models with multiple queues sharing a single server. For an overview of applications we refer to~\cite{boon2011applications} and for an overview of commonly used methods we refer to~\cite{vishnevskii2006mathematical}.

A general polling model has $n$ queues, each with a distinct arrival process (usually a Poisson process) with parameter $\lambda_i$, which are assumed to be independent from each other. Each queue has its own generally distributed service time from which is sampled independently. A single server is visiting each of the $n$ queues in a certain (possibly random) order to serve customers. After a certain period at a queue the server switches to the next queue. We assume that this switching takes zero time. Instead, we assume that if we switch to a queue that is non-empty, a setup is performed. Otherwise, we do not perform a setup and continue immediately to the next queue (see e.g.~\citealt{singh2002exact}). When all queues were empty before the arrival of a vehicle, we assume that a setup was started at the most recent departure epoch. This is a feature that has not been studied before in the polling literature, but that naturally represents the behaviour of our PFAs.

We will analyze the performance of PFAs regarding delay through polling models. Although we take a vertical queueing approach in those polling models (i.e. the vehicles are all stopped at the stop line at the intersection, occupying no space), the intersection control algorithm provides a one-to-one relation between the vertical queueing model and the PFAs. We visualize this in Figure~\ref{f:link}, where the black line represents a self-driving vehicle, and the red dotted line represents the corresponding `vehicle' in the vertical queueing model. Both `vehicles' enter the control region at the same time (so also the earliest possible arrival time at the intersection is the same for both). They also have the same service time, because as soon as the vehicles start to cross the intersection they have the same trajectory. So the delay for both vehicles, the difference between earliest possible crossing time and actual crossing time, is the same, as visualized in Figure~\ref{f:link}.

\begin{figure}[h!]
	\centering
	\includegraphics[width = 0.7\linewidth]{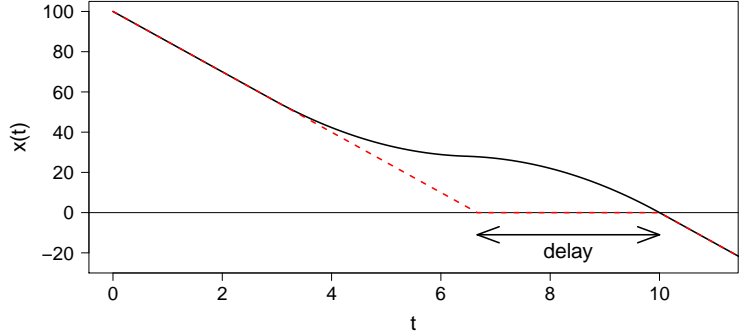}
	\caption{Visualization of the link between the traffic model with PFAs and polling models. The black line represents a self-driving vehicle, and the red dotted line represents the corresponding `vehicle' in the vertical queueing model.}
	\label{f:link}
\end{figure}

To make the connection between the traffic model and polling models more explicit, we argue how the traffic model translates to a polling model. The time $B$ in between vehicles from the same stream accessing the intersection is the service time in the polling model, whereas the clearance time $S$ is the setup time in the polling model. Which queue or lane is to be served is decided upon by the service discipline, respectively the PFA. %! Overbodig?

So, our intersection model precisely fits the framework of polling models. We will use the ideas and results already obtained for polling models to give a sound analysis of the traffic model discussed so far. From now on in this section, we will be focusing on the polling model and related results, therefore using queueing terminology.

\subsection{Mean Delay}

The specific assumptions result in a polling model that does not fall into the standard framework, and a fully analytical solution is difficult (if not impossible) to derive. So, we focus on approximations, being much faster and still accurate, and refrain from providing an analytical solution.

We focus on obtaining approximations for the mean delay that still require some analytical results, but that are easier to derive than the full distribution of the delay. We start with a definition of delay. The delay $D_i$ at lane $i$ is defined as the actual time of a car crossing the intersection minus the free-flow time in which a car could cross the intersection. $B_i$ denotes the service time of queue $i$, whereas $S_i$ denotes the setup time when we arrive at queue $i$. We have Poisson arrivals with rate $\lambda_i$ and define $\rho_i=\lambda_i\E[B_i]$ and $\rho=\sum_{i}\rho_i$, where $\rho$ is similar to the vehicle-to-capacity ratio. The approximations that we propose for the mean delay are all of the form,
\begin{equation}\label{eq:generalApproximation}
	\E[D_{i,app}^{P}]=\frac{K_{0,i}^{P}+K_{1,i}^{P}\rho+K_{2,i}^{P}\rho^2}{1-\rho},
\end{equation}
like in~\cite{boon2011closed}, where $K_{j,i}^{P}$ are constants that are yet to be determined and $P$ denotes the PFA. The constants, that might depend on $P$ and the arrival distribution (due to space limitations we only consider Poisson arrivals), are derived through requiring \eqref{eq:generalApproximation} to be exact in various limiting cases. These three cases are the following: \eqref{eq:generalApproximation} should match the mean delay for queue $i$ in the light-traffic limit, the derivative of the light-traffic limit and the heavy-traffic limit. Then we have a system of three equations with three unknowns, which we can solve to find the constants $K_{j,i}^P$. These approximations are based on the framework described in~\cite{boon2011closed}, which is in turn based on ideas developed in~\cite{reiman1988interpolation}. Note that \eqref{eq:generalApproximation} is only valid for $\rho < 1$, which is the condition for the polling model (and therefore also for our PFAs) to be stable.

We start with deriving the light-traffic limit for general service time and setup time distributions for the mean delay. The light-traffic here corresponds to the case where $$\P(\text{server not working and not setting up}) \uparrow 1,$$ which means that both $\lambda_i\E[B_i]$ and $\lambda_jE[S_i]$ should be close to zero. We denote with $X_{i}^{res}$ the residual or overshoot of the random variable $X$ with mean $\mathbb{E}[X^{res}]=\mathbb{E}[X^2]/(2\mathbb{E}[X]$).  Then we have the following lemma.

\begin{lemma}\label{t:lt_d}
The light-traffic limit for the mean delay, up to and including first-order terms, for all discussed PFAs, satisfies
\begin{equation}\label{eq:lt_d}
\E[D_{i}^{LT}] = \rho_i \E[B_{i}^{res}] +\sum_{j\neq i} \rho_j(\E[B_{j}^{res}]+\E[S_{i}]) +	\sum_{j\neq i} \lambda_j \E[S_i]\E[S_{i}^{res}].
\end{equation}
\end{lemma}

\begin{proof}
	We consider what happens in each phase of the cycle and argue what the waiting time is of a customer arriving at queue $i$.
	
	We have $n$ different visit periods, numbered $j=1,...,n$. If $j=i$, we only have to wait for a residual service time of the customer that is currently in service (using the PASTA property of Poisson arrivals). This happens with probability $\lambda_i\E[B_i]=\rho_i$. The contribution to the waiting time is thus $\rho_i\E[B_{i}^{res}]$. If $i\neq j$, we have to wait for the residual service time of the customer that is in service and for the setup time to our own queue $i$. This all happens with probability $\lambda_j\E[B_j]=\rho_j$, so the contribution to the waiting time is $\rho_j(\E[B_{j}^{res}]+\E[S_i])$.
	
	The setup periods: we again have $j=1,...,n$. The case $i=j$ does not occur, as we do not have a setup time in that case (we take the customer immediately into service). The cases $i\neq j$, occur with rate $\lambda_j\E[S_i]$ (which converges to zero) and if we arrive during such a period, we have to wait for a residual setup time. So the contribution is $\lambda_j\E[S_i]\E[S_{i}^{res}]$.

	Cases where we see more than one customer when we arrive in the system are all of order $\mathcal{O}(\rho^2)$ or higher, so we do not consider those terms.
	
	Summing all possibilities, we arrive at \eqref{eq:lt_d}.
\end{proof}

%\begin{remark}
%Note that as opposed to formula (3.11) in~\cite{boon2011closed}, where the light-traffic limit for a polling model with switchover times, that does not stop when all queues are empty, is derived, we do not have a constant term. This is the result of the residual structure in setup times that we assume. I.e. in light-traffic traffic, assuming residual setups, the mean delay for each queue is 0, which makes sense for our traffic model. Rewriting \eqref{eq:lt_d} to \begin{equation} \label{eq:lt_rew}
%	\E[D_{i}^{LT}] = \rho\E[B^{res}] + (\rho-\rho_i) \E[R_i]+ \frac{\E[S_{i}^{2}]}{2}\sum_{j\neq i} \lambda_j,
%\end{equation}
%reveals that some of the terms in (3.11) in~\cite{boon2011closed} simply cancel, because the waiting time during a switch-over time are not of $\mathcal{O}(1)$, but of order $\mathcal{O}(\rho)$. {\color{blue} check dit nog even voor de zekerheid: The remainder of the change from (3.11) in~\cite{boon2011closed} to \eqref{eq:lt_rew} can be clarified by the observation that we are dealing with (residual) setup times and that setup times are associated with the queue that is setting up.}
%\end{remark}

The heavy-traffic limit of the mean delay does depend on the PFA. In heavy traffic, the behaviour of our PFAs and regular polling models is the same. Consequently, the heavy-traffic limits for the exhaustive and gated PFAs are the same as the heavy-traffic limits for the exhaustive and gated disciplines in e.g.~\cite{boon2011polling}, where polling models with switch-over times (rather than setup times) are presented. Indeed, if the lengths of the setups and switch-overs are the same, the polling model with switch-overs (and without setup times) is the same as the polling model with setup times (but no switch-over times), because each setup will be performed in heavy traffic (as all queues are non-empty when the server visits them) and can be seen as an `ordinary' switch-over time. This implies that we can use the results from~\cite{boon2011polling}, so
\begin{equation}\label{eq:ht_w}
\E[D_{i}^{HT,P}]=\frac{\omega_{i}^P}{1-\rho}+o((1-\rho)^{-1}),
\end{equation}
with $P$ denoting the PFA, where
\begin{equation}\label{eq:omegaExh}
\omega_{i}^{exh} = \frac{1-\hat{\rho}_i}{2} \left(\frac{\sigma^2}{\sum_{j=1}^n \hat{\rho}_j(1-\hat{\rho}_j)}+\sum_{j=1}^n \E[S_j]\right),
\end{equation}
for the exhaustive PFA, with  $\sigma^2 = \E[B^2]/\E[B]$ (in case of Poisson arrivals) and $\hat{\rho}_i = \rho_i/\rho$ and
\begin{equation}\label{eq:omegaGat}
\omega_{i}^{gat} = \frac{1+\hat{\rho}_i}{2} \left(\frac{\sigma^2}{\sum_{j=1}^n \hat{\rho}_j(1+\hat{\rho}_j)}+\sum_{j=1}^n \E[S_j]\right)
\end{equation}
for the gated PFA.

The general approximation in \eqref{eq:generalApproximation} is now ready to be used. We obtain the following theorem.
\begin{theorem}
	The mean delay experienced for PFA $P$ can be approximated with Equation \eqref{eq:generalApproximation}, where
	\begin{align}
		&K_{0,i}^P = 0, \nonumber\\
		&K_{1,i}^P = \hat{\rho}_i \E[B_{i}^{res}] +\sum_{j\neq i} \hat{\rho}_j(\E[B_{j}^{res}]+\E[S_{i}]) +	\sum_{j\neq i} \hat{\lambda}_j\E[S_{i}^{res}] \E[S_i], \label{eq:constants}\\\nonumber
		&K_{2,i}^P = \omega_{i}^P - K_{1,i}^P,
	\end{align}
	with $\hat{\lambda}_i=\hat{\rho}_i/\E[B_i]$.
\end{theorem}

\begin{proof}
As mentioned before, we put three conditions on the constants $K_{j,i}^P$, $j=0,1,2$. These are the following
\begin{equation*}
\begin{aligned}
\E[D_{i,app}^{P}]\Big|_{\rho=0} &= \E[D_{i}^{LT}]\Big|_{\rho = 0},\\
\frac{d}{d\rho}\E[D_{i,app}^{P}]\Big|_{\rho=0} &= \frac{d}{d\rho}\E[D_{i}^{LT}]\Big|_{\rho = 0},\\
(1-\rho)\E[D_{i,app}^{P}]\Big|_{\rho\uparrow 1} &= \E[D_{i}^{HT,P}].
\end{aligned}
\end{equation*}

Using Lemma~\ref{t:lt_d} and Equation \eqref{eq:ht_w},

\begin{equation}\label{eq:setEquations}
\begin{aligned}
&K_{0,i}^P =  0,\\
&K_{0,i}^P+K_{1,i}^P = \hat{\rho}_i \E[B_{i}^{res}] +\sum_{j\neq i} \hat{\rho}_j(\E[B_{j}^{res}]+\E[S_{i}]) +	\sum_{j\neq i} \hat{\lambda}_j\E[S_{i}^{res}] \E[S_i],\\
&K_{0,i}^P+K_{1,i}^P+K_{2,i}^P = \E[D_{i}^{HT,P}]=\omega_{i}^P.
\end{aligned}
\end{equation}

It can easily be seen that \eqref{eq:setEquations} reduces to \eqref{eq:constants}.
\end{proof}

\begin{remark}
	The above mentioned results for mean delay can readily be extended to results for the mean number of vehicles in the queue, using Little's law. Together with the speed regulation algorithm, the physical length of the queue can be calculated (for example if we define the last vehicle that has already decelerated to be in the queue). This would give information about e.g. spillback of the intersection to other intersections.
\end{remark}

In general the approximations work fine for all discussed PFAs, as can be seen in Figure~\ref{f:meanDelay} (comparing the solid lines (the exact results) and the dashed lines (the approximations)). We present examples where we put $v_m=15$ m/s, $a_m=4$ m/s$^2$, $l=5$ m and $s=10$ m and where two lanes cross each other. We consider two cases where the load on both lanes is split differently: one case where $\rho_1=\rho_2$ (referred to as being symmetric) and one case where $\rho_1=3\rho_2$ (referred to as being asymmetric). Following~\cite{tachet2016revisiting}, we put $B=1$ s and $S=2.375$ s.  The two discussed PFAs result in the Figure~\ref{f:meanDelay}, where also, as a benchmark, the Batch Algorithm from~\cite{tachet2016revisiting} is considered, with a maximum batch size of 100. The approximations are also good for all other settings we simulated.

\begin{figure}[h!]
	\centering
	\includegraphics[width = 0.7\linewidth]{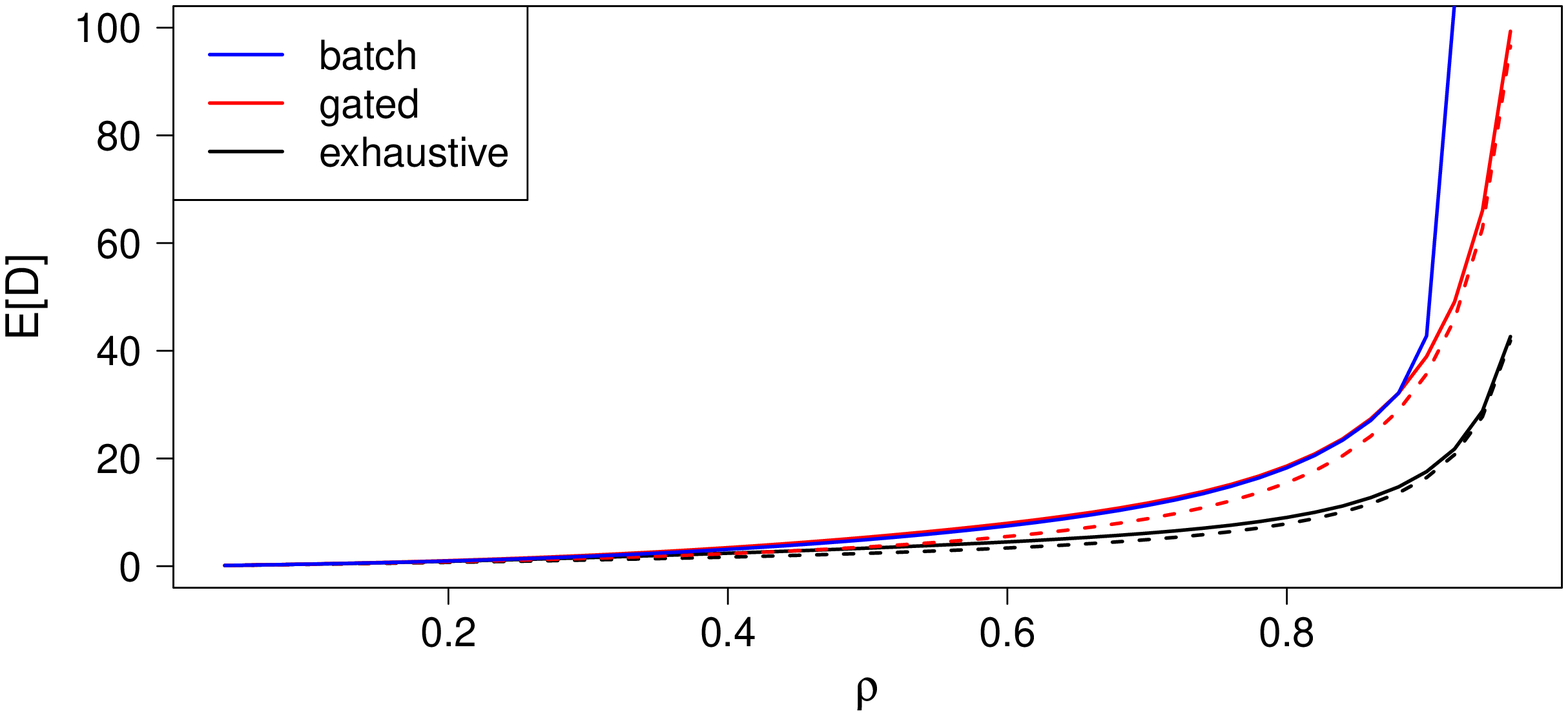}
	\vspace{0.5 cm}
	\includegraphics[width = 0.7\linewidth]{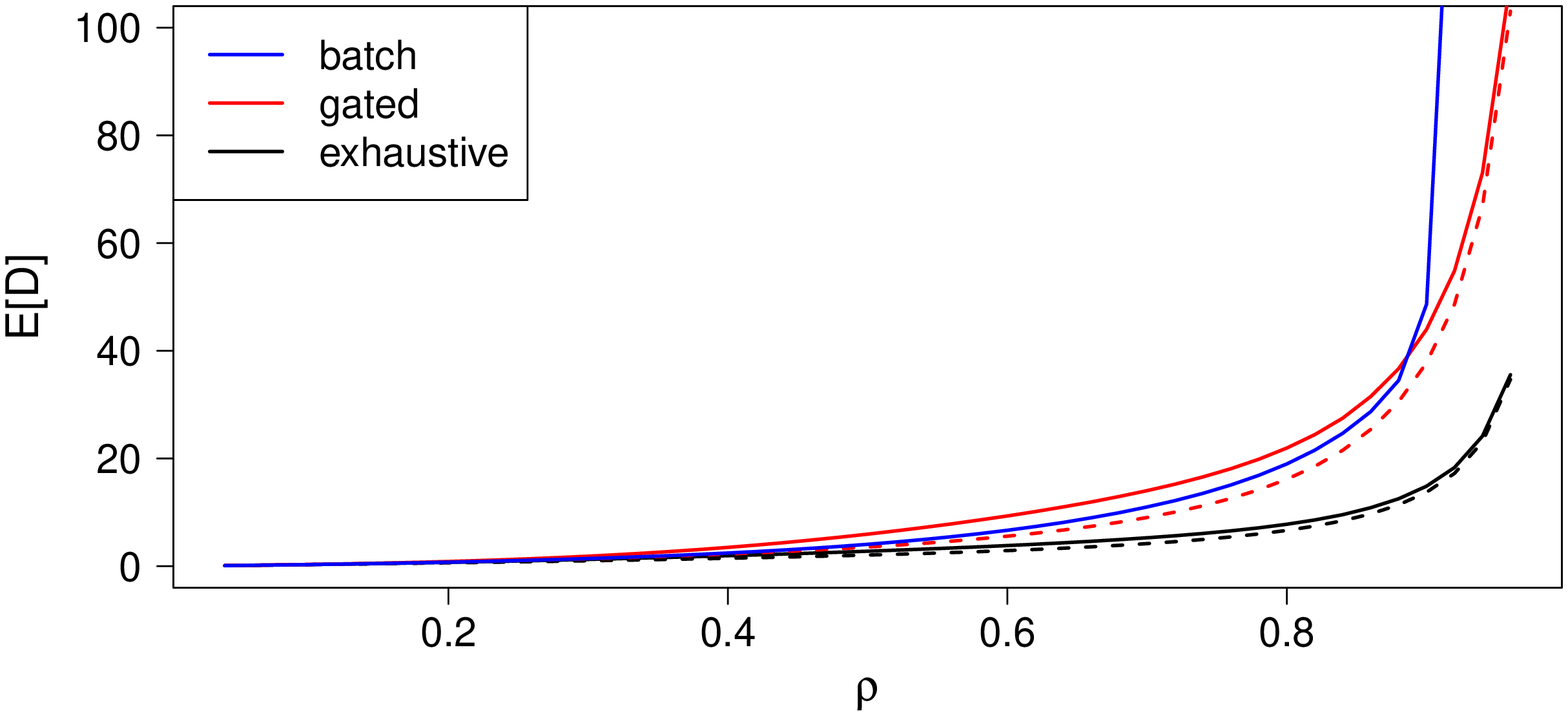}
	\caption{Mean delay experienced by an arbitrary car for the symmetric case (top) and asymmetric case (bottom). The solid lines represent simulation results and the dashed lines approximations.}
	\label{f:meanDelay}
\end{figure}
We see that the exhaustive PFA performs really well, if we focus on mean delay, compared to the other PFAs. This can also be understood from the heavy-traffic limits \eqref{eq:omegaExh} and \eqref{eq:omegaGat}. The performance of the Batch Algorithm is similar to that of the gated PFA, except for higher values of $\rho$, which is due to  the maximum batch size of 100. This maximum batch size causes a lower maximum capacity for the Batch Algorithm than for the exhaustive and gated PFA and therefore, the Batch Algorithm has a sharp increase in mean delay earlier than the other two PFAs. We expect the exhaustive PFA to be (very close to) optimal with respect to the mean delay. This optimality was, to some extent, already observed in e.g.~\cite{newell1969properties},~\cite{levy1990dominance} and~\cite{wu2013mathematical}.

\subsection{Fairness}

In order to show that the exhaustive PFA is not the best for all performance metrics we consider fairness in this subsection. We use the definition of fairness for polling models, denoted with $F$, as introduced in~\cite{shapira2015fairness},
\begin{equation*}
F = \frac{\E[N_{ahead}]}{\E[N_{total}]},
\end{equation*}
where $N_{ahead}$ denotes the number of cars an arbitrary car sees upon arrival and that are served ahead of it; and where $N_{total}$ denotes the total number of cars an arbitrary car sees upon arrival. In words this means that we quantify the percentage of cars that did not overtake an arbitrary car (on an intersection-wide basis).

We present simulation results for fairness for the same set of examples as for the mean delay.
\begin{figure}[h!]
	\centering
	\includegraphics[width = 0.7\linewidth]{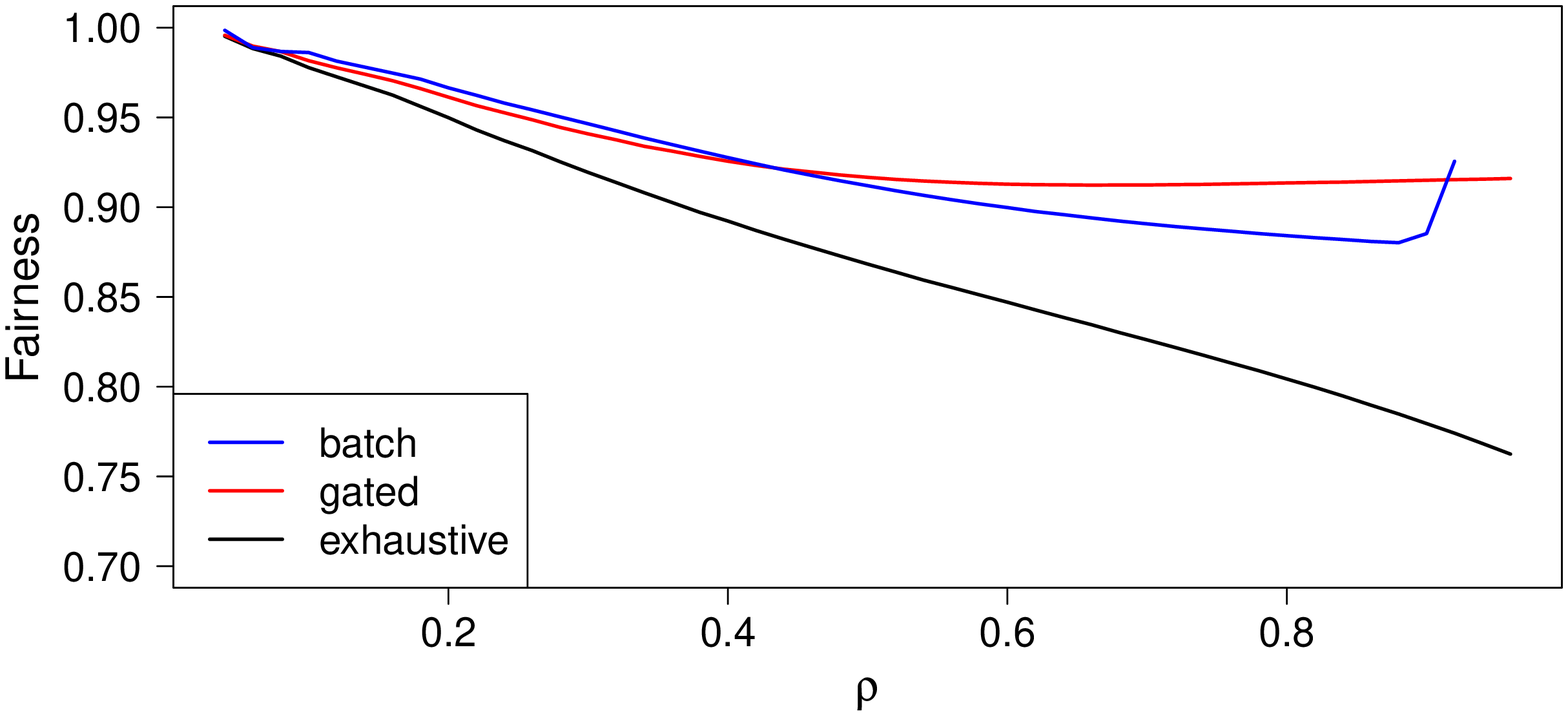}
	\vspace{0.5 cm}
	\includegraphics[width = 0.7\linewidth]{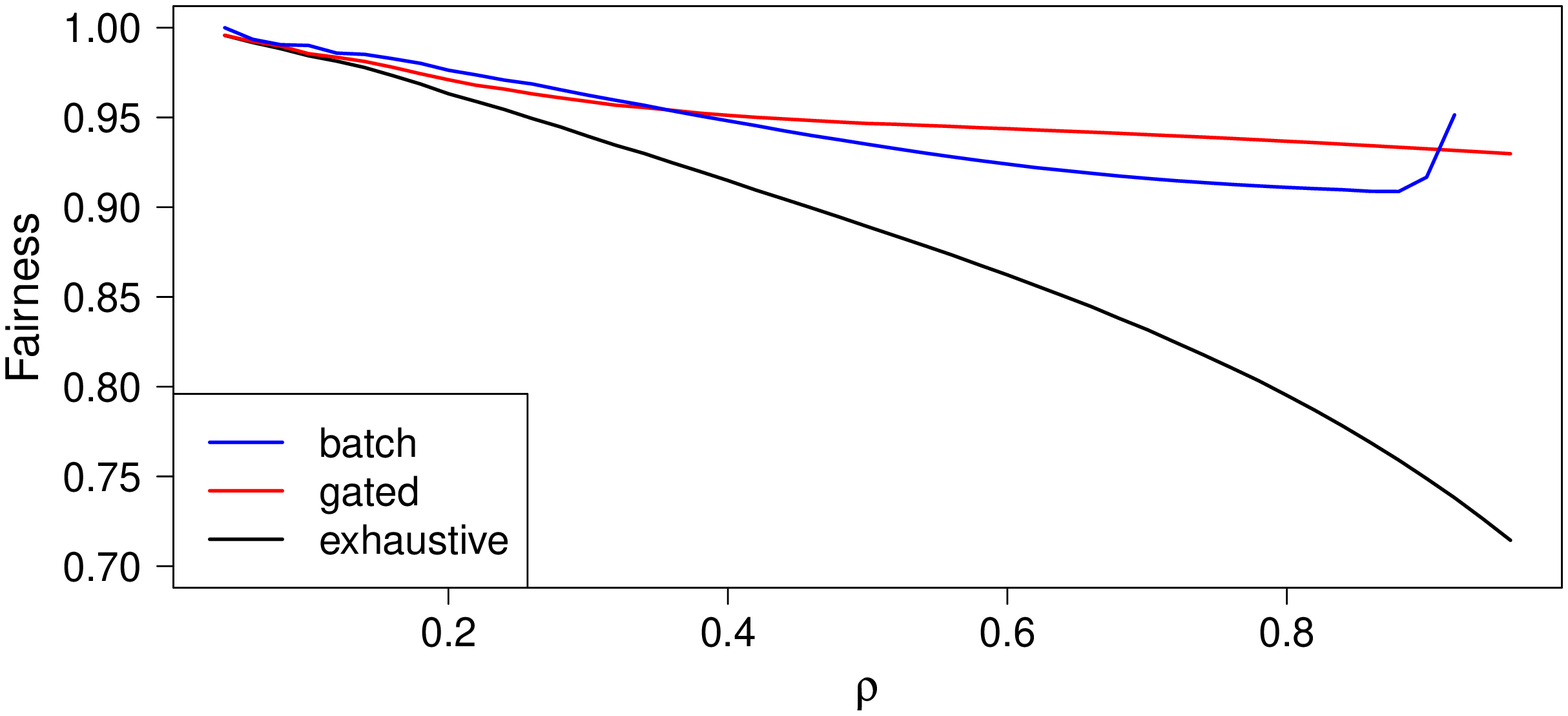}
	\caption{Fairness experienced by an arbitrary car for the symmetric case (top) and asymmetric case (bottom).}
	\label{f:fairness}
\end{figure}

Considering fairness, we see once more that the gated PFA is close to the Batch Algorithm for values of $\rho$ that are not too high. The increase of fairness for high values of $\rho$ for the Batch Algorithm is due to the maximum batch size of 100. The exhaustive PFA is worse on fairness, but is still above 75\%. It seems that a low mean delay results in a relatively low fairness, showing a potential need to balance the two performance measures, which is to some extent visible in the increase of fairness for the Batch Algorithm and high values of $\rho$.

\section{Comparison traditional Traffic Light and PFAs} \label{h:comparison}

The goal of this section is to provide a comparison between traditional traffic lights and PFAs on basis of delay. As a measure for the traditional traffic light we use the traffic simulator SUMO. We will consider two scenarios in SUMO: one with fixed control and one with adaptive control (based on so-called time loss in the SUMO User Documentation). We will compare these two scenarios with the exhaustive PFA.

We again consider two examples where two lanes cross each other and the vehicle to capacity ratio is in the first example the same on both lanes, whereas in the second example the ratio between the loads on the lanes is $1:3$. For the exhaustive PFA we again put $B=1$ s and $S=2.375$ s. For the fixed control simulation in SUMO and the first example we assume a green period for both lanes of $22$ s and an amber period of $3$ s; for the second example we pick green periods of $11$ and $33$ s and amber periods of length $3$ s. Note that some of the results for the fixed control in Figure~\ref{f:reeel} could be slightly improved by adapting the length of the green period. For the adaptive control in SUMO we assume a maximum green period duration of $45$ s and an amber period of $3$ s for the symmetric example and a maximum green of $22$ and $68$ s for the asymmetric case. Note that we do not have to define the variable $B$ in SUMO, as the vehicles themselves will decide what $B$ is. The delay in SUMO for the fixed and adaptive control is obtained in the following way: we compute the mean time spent in the system for all vehicles and subtract the mean time vehicles spent in the system under free-flow conditions (i.e. putting the traffic light at green for one lane all the time). We take exactly the same arrivals for all three scenarios.

\begin{figure}[h!]
	\centering
	\includegraphics[width = 0.7\linewidth]{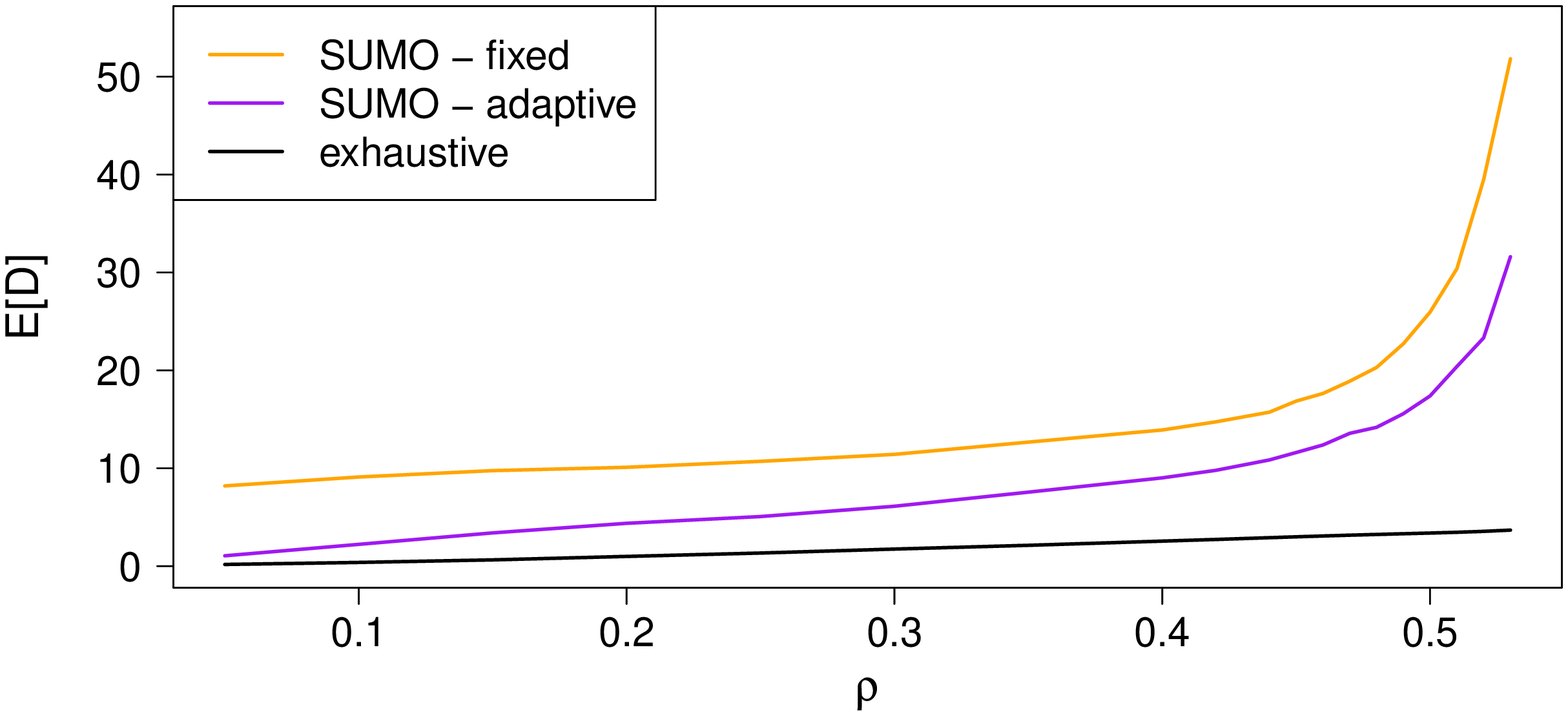}
	\vspace{0.5 cm}
	\includegraphics[width = 0.7\linewidth]{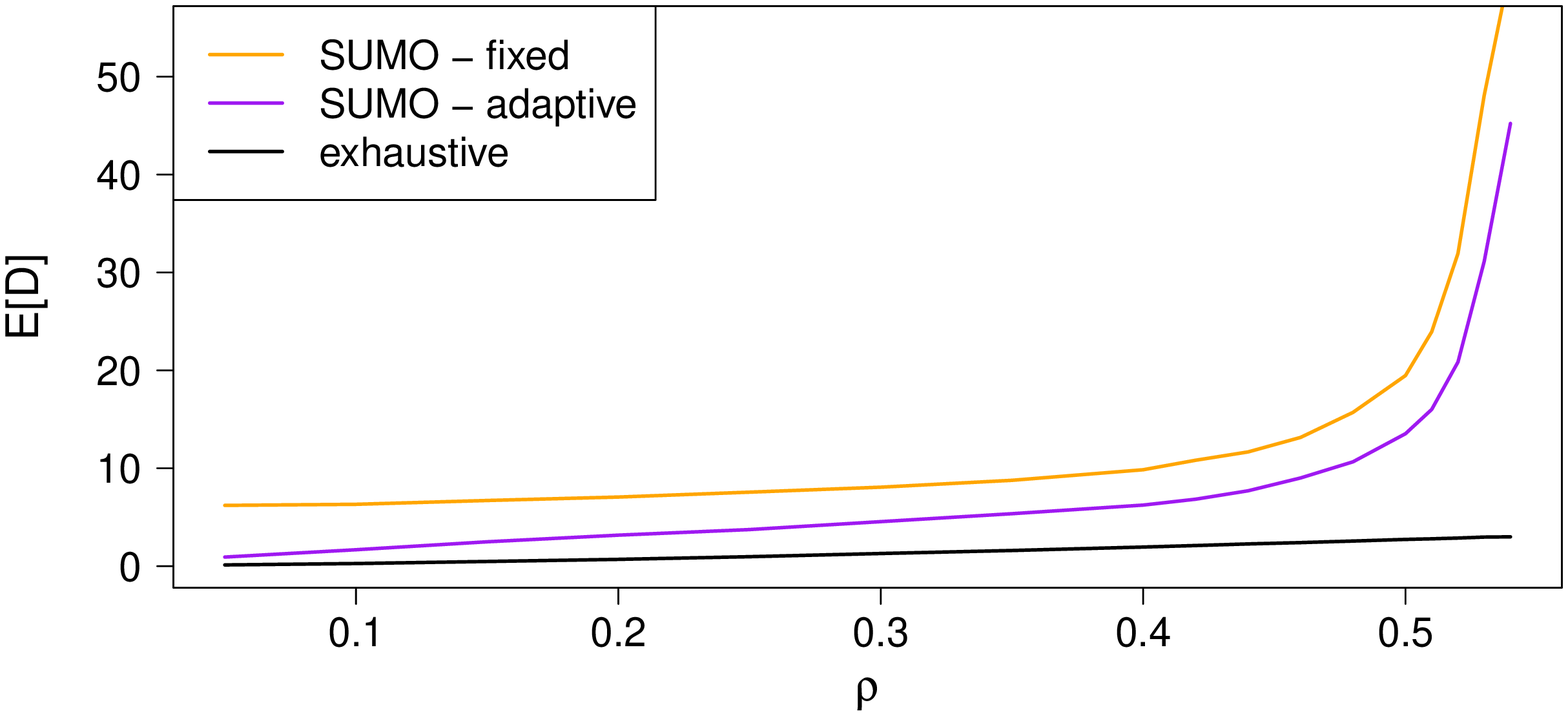}
	\caption{Mean delay for an arbitrary car for traditional traffic lights (represented by SUMO) and the exhaustive PFA for the symmetric example (top) and the asymmetric case (bottom).}
	\label{f:reeel}
\end{figure}

We see in Figure~\ref{f:reeel} that there is quite a difference between either the fixed cycle traffic light or the adaptive traffic light, and the exhaustive PFA. To some extent, this was also observed in~\cite{tachet2016revisiting}. The capacity of the intersection for the latter case is almost twice as high as for the traditional traffic light, showing a huge potential in resolving congestion. This is mainly due to the speed regulation of vehicles, which increases the speed of vehicles crossing the intersection, but also due to the scheduling strategies of the PFA.

\section{Conclusion and Discussion} \label{h:conclusion}

We have shown that significant gains can be obtained compared to nowadays traffic when speed regulation and PFAs can be employed and have given ways to decrease mean delay on intersections. This has been shown through a connection between polling models and PFAs.

It seems that the exhaustive PFA is close to optimality with respect to mean delay. However, the exhaustive PFA exhibits relatively poor fairness characteristics. It might be worthwhile to find a balance between mean delay and (e.g.) fairness in order to obtain some kind of optimal setting for the PFA. A possibility hereto might be the so-called $k$-limited discipline in polling models, where for each lane an upper bound to the platoon size is set. Intuitively, the $k$-limited discipline is similar to the exhaustive discipline, except for this maximum size of the platoon.

In principle our PFAs could be used in nowadays traffic as well. The only requirement is that it must be known on an intersection wide basis in which order the vehicles arrive. The requirement that we can control the speed of arriving vehicles is not needed to execute the PFAs. This assumption would only play a role in what the variables $B$ and $S$ would look like. But even then, the scheduling part of a PFA might still be used. Using some kind of speed advisory system for conventional vehicles, it might be possible to come close to the performance of the PFAs based on self-driving vehicles.

A future direction of research is to investigate more realistic intersection scenarios, yet we expect similar results. Depending on the extension, our results readily apply, if at most one stream of vehicles is allowed to cross the intersection, or need to be generalized. We also would like to extend our approximations to obtain analytical results for fairness.

We have studied an isolated intersection, where vehicles arrive individually in the control region. In a network of intersections there are several complications. Firstly, the arrival processes of vehicles become dependent. Moreover, the interplay between various intersections is non-trivial. Already for a tandem of fixed cycle traffic light intersections, it is difficult to find a good green wave, see e.g.~\cite{oblakova2017green}. Our PFAs are much less strict on e.g. the cycle length, imposing an even more difficult task of balancing a whole network of intersections. Once more, the $k$-limited PFA (having a fixed maximum cycle length) might prove to be an outcome in this respect.

A study on how realistic our proposed models are, might also be relevant. We assume e.g. that each vehicle is able to perfectly match the criteria we set in the speed regulation assumptions. For example, there might be some uncertainty in the control of a self-driving vehicle. A notion like string-stability of a platoon of vehicles (see e.g.~\citealt{swaroop1996string}) might be investigated for our proposed models.

\section*{Acknowledgments}

We would like to thank Johan van Leeuwaarden, Onno Boxma, Wim van Nifterick and Serge Hoogendoorn for interesting discussions.

\section*{Funding}

This work was supported by NWO under Grant 438-13-206.

\bibliographystyle{tfcad}
\bibliography{fctl}

\begin{thebibliography}{31}
\newcommand{\enquote}[1]{``#1''}
\providecommand{\natexlab}[1]{#1}
\providecommand{\url}[1]{\normalfont{#1}}
\providecommand{\urlprefix}{}

\bibitem[Boon(2011)]{boon2011polling}
Boon, M.A.A. 2011. ``Polling models: from theory to traffic intersections.''
  PhD thesis, Eindhoven University of Technology.

\bibitem[Boon, van~der Mei, and Winands(2011)]{boon2011applications}
Boon, M.A.A., R.D. van~der Mei, and E.M.M. Winands. 2011. ``Applications of
  polling systems.'' \emph{{S}urveys in {O}perations {R}esearch and
  {M}anagement {S}cience} 16 (2): 67--82.

\bibitem[Boon et~al.(2011)]{boon2011closed}
Boon, M.A.A., E.M.M. Winands, I.J.B.F. Adan, and A.C.C. van Wijk. 2011.
  ``Closed-form waiting time approximations for polling systems.''
  \emph{Performance Evaluation} 68 (3): 290--306.

\bibitem[Darroch(1964)]{darroch}
Darroch, J.N. 1964. ``On the traffic light queue.'' \emph{The Annals of
  Mathetical Statistics} 35: 380--388.

\bibitem[Dib et~al.(2014)]{dib2014optimal}
Dib, W., A.~Chasse, P.~Moulin, A.~Sciarretta, and G.~Corde. 2014. ``Optimal
  energy management for an electric vehicle in eco-driving applications.''
  \emph{Control Engineering Practice} 29: 299--307.

\bibitem[Dresner and Stone(2008)]{dresner2008multiagent}
Dresner, K., and P.~Stone. 2008. ``A multiagent approach to autonomous
  intersection management.'' \emph{Journal of Artificial Intelligence Research}
  31: 591--656.

\bibitem[Helbing, Farkas, and Vicsek(2000)]{helbing2000simulating}
Helbing, D., I.~Farkas, and T.~Vicsek. 2000. ``Simulating dynamical features of
  escape panic.'' \emph{Nature} 407: 487--490.

\bibitem[Helbing and Mazloumian(2009)]{helbing2009operation}
Helbing, D., and A.~Mazloumian. 2009. ``Operation regimes and slower-is-faster
  effect in the control of traffic intersections.'' \emph{The European Physical
  Journal B-Condensed Matter and Complex Systems} 70 (2): 257--274.

\bibitem[Kockelkoren(2018)]{kockelkoren2018}
Kockelkoren, L.M.C. 2018. ``{C}entralized merge control for {FLEET}, a material
  handling {AGV} system.'' Master's thesis, {E}indhoven University of
  {T}echnology.

\bibitem[Lawitzky, Wollherr, and Buss(2013)]{lawitzky2013energy}
Lawitzky, A., D.~Wollherr, and M.~Buss. 2013. ``Energy optimal control to
  approach traffic lights.'' In \emph{Intelligent Robots and Systems (IROS),
  2013 IEEE/RSJ International Conference on}, 4382--4387. IEEE.

\bibitem[Levy, Sidi, and Boxma(1990)]{levy1990dominance}
Levy, H., M.~Sidi, and O.J. Boxma. 1990. ``Dominance relations in polling
  systems.'' \emph{Queueing Systems} 6 (1): 155--171.

\bibitem[Liu, Wang, and Hoogendoorn(2019)]{liuwanghoogendoorn2019}
Liu, M., M.~Wang, and S.~Hoogendoorn. 2019. ``Optimal platoon trajectory
  planning approach at arterials.'' \emph{Transportation Research Record}
  \urlprefix\url{https://doi.org/10.1177/0361198119847474}.

\bibitem[Miculescu and Karaman(2014)]{miculescu2014polling}
Miculescu, D., and S.~Karaman. 2014. ``Polling-systems-based control of
  high-performance provably-safe autonomous intersections.'' In \emph{IEEE 53rd
  Annual Conference on Decision and Control (CDC)}, 1417--1423. IEEE.

\bibitem[Miculescu and Karaman(2016)]{miculescu2016polling}
Miculescu, D., and S.~Karaman. 2016. ``Polling-systems-based autonomous vehicle
  coordination in traffic intersections with no traffic signals.'' \emph{arXiv
  preprint:1607.07896} .

\bibitem[Milan{\'e}s et~al.(2010)]{milanes2010controller}
Milan{\'e}s, V., J.~P{\'e}rez, E.~Onieva, and C.~Gonz{\'a}lez. 2010.
  ``Controller for urban intersections based on wireless communications and
  fuzzy logic.'' \emph{IEEE Transactions on Intelligent Transportation Systems}
  11 (1): 243--248.

\bibitem[Newell(1965)]{newell65}
Newell, G.~F. 1965. ``Approximation methods for queues with application to the
  fixed-cycle traffic light.'' \emph{SIAM Review} 7 (2): 223--240.

\bibitem[Newell(1969)]{newell1969properties}
Newell, G.F. 1969. ``Properties of vehicle-actuated signals: I. one-way
  streets.'' \emph{Transportation Science} 3 (1): 30--52.

\bibitem[Oblakova et~al.(2017)]{oblakova2017green}
Oblakova, A., A.~Al~Hanbali, R.J. Boucherie, and J.C.W. van Ommeren. 2017.
  ``Green wave analysis in a tandem of traffic-light intersections.'' .

\bibitem[Papageorgiou et~al.(2003)]{papageorgiou2003review}
Papageorgiou, M., C.~Diakaki, V.~Dinopoulou, A.~Kotsialos, and Y.~Wang. 2003.
  ``Review of road traffic control strategies.'' \emph{Proceedings of the IEEE}
  91 (12): 2043--2067.

\bibitem[Rafaeli, Barron, and Haber(2002)]{rafaeli2002effects}
Rafaeli, A., G.~Barron, and K.~Haber. 2002. ``The effects of queue structure on
  attitudes.'' \emph{Journal of Service Research} 5 (2): 125--139.

\bibitem[Reiman and Simon(1988)]{reiman1988interpolation}
Reiman, M.I., and B.~Simon. 1988. ``An interpolation approximation for queueing
  systems with {P}oisson input.'' \emph{Operations Research} 36 (3): 454--469.

\bibitem[Resing(1993)]{resing1993polling}
Resing, J.A.C. 1993. ``Polling systems and multitype branching processes.''
  \emph{Queueing Systems} 13 (4): 409--426.

\bibitem[Rios-Torres and Malikopoulos(2017)]{rios2017survey}
Rios-Torres, J., and A.A. Malikopoulos. 2017. ``A survey on the coordination of
  connected and automated vehicles at intersections and merging at highway
  on-ramps.'' \emph{IEEE Transactions on Intelligent Transportation Systems} 18
  (5): 1066--1077.

\bibitem[Shapira and Levy(2016)]{shapira2015fairness}
Shapira, G., and H.~Levy. 2016. ``On fairness in polling systems.''
  \emph{Annals of Operations Research}
  \urlprefix\url{https://doi.org/10.1007/s10479-016-2247-8}.

\bibitem[Singh and Srinivasan(2002)]{singh2002exact}
Singh, M.P., and M.M. Srinivasan. 2002. ``Exact analysis of the state-dependent
  polling model.'' \emph{Queueing Systems} 41 (4): 371--399.

\bibitem[Swaroop and Hedrick(1996)]{swaroop1996string}
Swaroop, D., and J.K. Hedrick. 1996. ``String stability of interconnected
  systems.'' \emph{IEEE {T}ransactions on {A}utomatic {C}ontrol} 41 (3):
  349--357.

\bibitem[Tachet et~al.(2016)]{tachet2016revisiting}
Tachet, R., P.~Santi, S.~Sobolevsky, L.~I. Reyes-Castro, E.~Frazzoli,
  D.~Helbing, and C.~Ratti. 2016. ``Revisiting street intersections using
  slot-based systems.'' \emph{PloS {ONE}} 11 (3): e0149607.
  \urlprefix\url{https://doi.org/10.1371/journal.pone.0149607}.

\bibitem[Tanner(1962)]{tanner62}
Tanner, J.~C. 1962. ``A theoretical analysis of delays at an uncontrolled
  intersection.'' \emph{Biometrika} 49 (1/2): 163--170.

\bibitem[van Leeuwaarden(2006)]{fctlsolo}
van Leeuwaarden, J. S.~H. 2006. ``Delay analysis for the fixed-cycle
  traffic-light queue.'' \emph{Transportation Science} 40 (2): 189--199.

\bibitem[Vishnevskii and Semenova(2006)]{vishnevskii2006mathematical}
Vishnevskii, V.M., and O.V. Semenova. 2006. ``Mathematical methods to study the
  polling systems.'' \emph{Automation and Remote Control} 67 (2): 173--220.

\bibitem[Wu, Yan, and Abbas-Turki(2013)]{wu2013mathematical}
Wu, J., F.~Yan, and A.~Abbas-Turki. 2013. ``Mathematical proof of effectiveness
  of platoon-based traffic control at intersections.'' In \emph{Intelligent
  Transportation Systems-(ITSC), 2013 16th International IEEE Conference on},
  720--725. IEEE.

\end{thebibliography}

\end{document}